%% file: shortgame2.tex
\documentclass[12pt]{amsart}

\input{header/ams_header.tex}

\newcommand{\norm}[1]{\left\lVert#1\right\rVert}

\renewcommand{\Re}{\operatorname{Re}}
\renewcommand{\Im}{\operatorname{Im}}
\DeclareMathOperator{\tr}{tr}

\DeclareMathOperator{\sgn}{sgn}
\DeclareMathOperator{\HLPC}{\Gamma_0}
\DeclareMathOperator{\EHLPC}{E\Gamma_0}
\usepackage{braket}

\title{The set of quantum correlations is not closed}

\author{William Slofstra}

\begin{document}

\begin{abstract}
    We construct a linear system non-local game which can be played perfectly
    using a limit of finite-dimensional quantum strategies, but which cannot be
    played perfectly on any finite-dimensional Hilbert space, or even with any
    tensor-product strategy. In particular, this shows that the set of
    (tensor-product) quantum correlations is not closed. The constructed
    non-local game provides another counterexample to the ``middle'' Tsirelson
    problem, with a shorter proof than our previous paper (though at the loss
    of the universal embedding theorem). We also show that it is undecidable to
    determine if a linear system game can be played perfectly with a
    finite-dimensional strategy, or a limit of finite-dimensional quantum
    strategies. 
\end{abstract}

\maketitle

\section{Introduction}

A two-player non-local game $\mcG$ consists of finite question sets $\mcI_A$
and $\mcI_B$, finite output sets $\mcO_A$ and $\mcO_B$, and a function $V :
\mcO_A \times \mcO_B \times \mcI_A \times \mcI_B \arr \{0,1\}$.  During the
game, the two players, commonly called Alice and Bob, are given inputs $x \in
\mcI_A$ and $y \in \mcI_B$ respectively, and return outputs $a \in \mcO_A$ and
$b \in \mcO_B$ respectively. The players win if $V(a,b|x,y) = 1$, and lose if
$V(a,b|x,y)=0$. The players know the rules of the game, and can decide ahead of
time on their strategy. However, once the game is in progress, they are unable
to communicate, meaning they do not know each others inputs or subsequent
choices. This can make it impossible for the players to win with certainty.

Imagine that the game is played repeatedly. To an outside observer, Alice and
Bob's actions during the game are described by the probability $p(a,b|x,y)$
that Alice and Bob output $a \in \mcO_A$ and $b \in \mcO_B$ on inputs $x \in
\mcI_A$ and $y \in \mcI_B$. The collection $\{p(a,b|x,y)\} \subset \R^{\mcO_A
\times \mcO_B \times \mcI_A \times \mcI_B}$ is called a \emph{correlation
matrix} (or a \emph{behaviour}).  Which correlation matrices can be achieved
depends on the physical model. For instance, a correlation matrix
$\{p(a,b|x,y)\}$ is said to be \emph{classical} if it can be achieved using
classical shared randomness. Formally, this means that there must be some
integer $k \geq 1$, a probability distribution
$\{\lambda_i \}$ on $\{1,\ldots,k\}$, probability distributions $\{p^{ix}_a\}$
on $\mcO_A$ for each $1 \leq i \leq k$ and $x \in \mcI_A$, and probability
distributions $\{q^{iy}_b\}$ on $\mcO_B$ for each $1 \leq i \leq k$ and $y \in
\mcI_B$, such that 
\begin{equation*}
    p(a,b|x,y) = \sum_{i=1}^k \lambda_i p^{ix}_a q^{iy}_b \text{ for all }
        (a,b,x,y) \in \mcO_A \times \mcO_B \times \mcI_A \times \mcI_B.
\end{equation*}
The set of classical correlation matrices is denoted by
$C_c(\mcO_A,\mcO_B,\mcI_A,\mcI_B)$, although we typically write $C_c$ when
the output and input sets are clear. 

In quantum information, we are interested in what correlations can be achieved
with a shared quantum state. Accordingly, a correlation matrix is said to be
\emph{quantum} if there are finite-dimensional Hilbert spaces $H_A$ and $H_B$,
a quantum state $\ket{\psi} \in H_A \otimes H_B$, projective
measurements\footnote{A projective measurement on a Hilbert space $H$ is a
collection $\{P_x\}_{x \in X}$ of self-adjoint operators on $H$, such that
$P_x^2 = P_x$ for all $x \in X$, and $\sum_{x \in X} P_x = \Id$. The set
$X$ is interpreted as the set of measurement outcomes.} $\{M^{x}_a\}_{a \in
\mcO_A}$ on $H_A$ for every $x \in \mcI_A$, and projective measurements
$\{N^{y}_b\}_{b \in \mcO_B}$ on $H_B$ for every $y \in \mcI_B$, such that
\begin{equation*}
    p(a,b|x,y) = \bra{\psi} M^x_{a} \otimes N^y_{b} \ket{\psi} \text{ for all } 
        (a,b,x,y) \in \mcO_A \times \mcO_B \times \mcI_A \times \mcI_B.
\end{equation*}
The set of quantum correlation matrices is denoted by $C_q \cong
C_q(\mcO_A,\mcO_B,\mcI_A,\mcI_B)$. There are two natural variations on this
definition. We can drop the requirement that $H_A$ and $H_B$ be
finite-dimensional, in which case we get another set of correlations often
denoted by $C_{qs}$. We can also look at correlations which can be realized as
limits of finite-dimensional quantum correlations; the corresponding
correlation set is the closure of $C_q$, and is typically denoted by $C_{qa}$.
It is well-known that $C_{qs} \subseteq C_{qa}$, and consequently $C_{qa}$ is
also the closure of $C_{qs}$ \cite{SW08}.

Since $C_{qs} \subseteq C_{qa}$, we get a hierarchy of correlation sets
\begin{equation*}
    C_{c} \subseteq C_{q} \subseteq C_{qs} \subseteq C_{qa}.
\end{equation*}
All the sets involved are convex, and $C_c$ and $C_{qa}$ are both closed. 
Bell's celebrated theorem \cite{Be64} states that $C_c \neq C_q$, and
furthermore that the two sets can be separated by a hyperplane. It has been a
longstanding open problem to determine the relationship between the 
quantum correlation sets, and in particular to determine whether $C_q$ and
$C_{qs}$ are closed (see, i.e., \cite{Ts06,WCD08,Fr12,BLP17}). Part of the
interest in this latter question comes from the resource theory of non-local
games: $C_q \neq C_{qa}$ if and only if there is a non-local game which can be
played optimally (with respect to some payoff function) using a limit of
finite-dimensional quantum strategies, but cannot be played optimally using any
fixed dimension. Numerical evidence has suggested that even very simple
non-local games might have this property \cite{PV10, LW}. For variants of non-local
games (for instance, with quantum questions, or infinite output sets), there
are several examples of games with this property \cite{LTW13,MV14,RV15}.

The purpose of this paper is to show that there are indeed non-local games
(with finite classical input and output sets) that cannot be played optimally
using any fixed dimension. A \emph{perfect strategy} for a non-local game
$\mcG$ is a correlation matrix $\{p(a,b|x,y)\}$ such that Alice and Bob win
with probability one on every pair of inputs $x$ and $y$. Formally, this means
that for all $(a,b,x,y) \in \mcO_A \times \mcO_B \times \mcI_A \times \mcI_B$,
if $V(a,b|x,y) = 0$, then $p(a,b|x,y)=0$. 
\begin{theorem}\label{T:main}
    There is a non-local game with a perfect strategy in $C_{qa}$, but no
    perfect strategy in $C_{qs}$. 
\end{theorem}
In particular, neither $C_q$ or $C_{qs}$ are closed. The proof is constructive,
with the game in question having input sets of size $184$ and
$235$, and output sets of size $8$ and $2$. 

The set $C_q$ is related to the cone of completely positive-semidefinite (cpsd)
matrices defined in \cite{LP15}. An $n \times n$ matrix $M$ is said to be cpsd
if there are non-negative operators $P_1,\ldots,P_n$ on some finite-dimensional
Hilbert space with $M_{ij} = \tr(P_i P_j)$ for all $1 \leq i,j \leq n$. By a
theorem of Sikora and Varvitsiotis \cite{SV16}, the set $C_{q}$ is an affine
slice of the cone of cpsd matrices, so the cone of cpsd matrices is not
closed as a consequence of Theorem \ref{T:main}. 

The fact that $C_{qs} \neq C_{qa}$ also has an interesting reformulation. Let
$G_i$ be the $n$-fold free product $\Z_{m} * \cdots * \Z_{m}$, where
$n=|\mcI_i|$ and $m = |\mcO_i|$, for $i=A,B$. Let $M^{x}_a$ denote the $a$th
spectral projector of the $x$th factor of $G_A$ in the full group $C^*$-algebra
$C^*(G_A)$ of $G_A$, and define $M^y_b$ similarly for $C^*(G_B)$. For each
$i=A,B$, find a faithful representation $\nu_i$ of $C^*(G_i)$ on some Hilbert
space $H_i$. The minimal (or spatial) tensor product $C^*(G_A) \otimes_{s}
C^*(G_B)$ is the norm-closure of the image $\nu_A(C^*(G_A)) \otimes
\nu_B(C^*(G_B))$ in the $C^*$-algebra $\mcB(H_A \otimes H_B)$. A correlation
matrix $\{p(a,b|x,y)\}$ belongs to $C_{qa}$ if and only if there is a state
$\omega$ on the $C^*$-algebra $C^*(G_A) \otimes_{s} C^*(G_B)$ with 
\begin{equation*}
    p(a,b|x,y) = \omega(M^x_a \otimes N^y_b)
\end{equation*}
for all $(a,b,x,y)\in \mcO_A \times \mcO_B \times \mcI_A \times \mcI_B$
\cite{SW08,Fr12}. On the other hand, the correlation matrix belongs to $C_{qs}$ if and
only if there are representations $\phi_i$ of $G_i$ on $H_i$, $i=A,B$, and a
vector state $\ket{\psi} \in H_A \otimes H_B$, with
\begin{equation*}
    p(a,b|x,y) = \bra{\psi} \phi_A(M^{x}_a) \otimes \phi_B(N^y_b) \ket{\psi}
\end{equation*}
for all $(a,b,x,y) \in \mcO_A \times \mcO_B \times \mcI_A \times \mcI_B$. Since
$C_{qs} \neq C_{qa}$, there can be states on the minimal tensor product
$C^*(G_A) \otimes_{s} C^*(G_B)$ which do not come from vector states on some
tensor-product $\phi_A \otimes \phi_B$ of representations $\phi_A$ and
$\phi_B$.

There is another candidate set of quantum correlations, the commuting-operator
correlations $C_{qc}$, which contains $C_{qa}$. Determining whether $C_{qc}$ is
known to be equal to $C_{t}$ for any $t \in {q,qs,qa}$ is known as Tsirelson's
problem \cite{Ts06,DP16}. In a previous paper \cite{Sl16}, we showed that
$C_{qs} \neq C_{qc}$. By showing that $C_{qs} \neq C_{qa}$, we provide another
proof of this fact. The proof that $C_{qs} \neq C_{qc}$ in \cite{Sl16} uses a
universal embedding theorem, which states that every finitely-presented group
embeds in the solution group of a linear system game. In this paper, we follow
a similar line, proving a restricted embedding theorem for a subclass of
finitely-presented groups which we call linear-plus-conjugacy groups. For the
proof of this restricted embedding theorem, we use a completely different
method from \cite{Sl16}, with the result that the proof is much shorter.
However, it remains an open problem to prove the universal embedding theorem 
via the new approach.

An easy consequence of the universal embedding theorem is that it is
undecidable to determine if a linear system game has a perfect strategy in
$C_{qc}$. In this paper we prove a stronger result by applying our restricted
embedding theorem to Kharlampovich's example \cite{Kha81} of a finitely
presented solvable group with an undecidable word problem. 
\begin{theorem}\label{T:main2}
    There is a (recursive) family of linear system games such that
    \begin{enumerate}[(a)]
        \item it is undecidable to determine if a game in the family has a 
            perfect strategy in $C_{qa}$, and
        \item every game in the family has a perfect strategy in $C_{qc}$ if and only
            if it has a perfect strategy in $C_{qa}$. 
    \end{enumerate}
\end{theorem}
Kharlampovich's construction has been extended by Kharlampovich, Myasnikov, and
Sapir to show that the word problem for finitely-presented residually-finite
groups can be as hard as any computable function \cite{KMS17}.\footnote{The
word problem for finitely-presented residually-finite groups is always
decidable, so this is the best possible lower bound.} Using this extension, we
can show:
\begin{theorem}\label{T:main3}
    Let $f : \mbN \arr \mbN$ be a computable function. Then there is a family
    of linear system games $\mcG_n$, $n \in \mbN$, such that
    \begin{enumerate}[(a)]
        \item the games $\mcG_n$ have input and output sets of size $\exp(O(n))$,
            and the function $n \mapsto \mcG_n$ is computable in $\exp(O(n))$-time;
        \item for any algorithm accepting the language 
            \begin{equation*}
                \{ n \in \mbN : \mcG_n \text{ has a perfect strategy in } C_{q} \},
            \end{equation*}
            the maximum running time over inputs $n \leq N$ is at least $f(N)$
            when $N$ is sufficiently large;
        \item $\mcG_n$ has a perfect strategy in $C_{qc}$ if and only if it has
            a perfect strategy in $C_{q}$. 
    \end{enumerate}
\end{theorem}
Theorem \ref{T:main3} has the following corollary.
\begin{cor}\label{C:main3}
    It is undecidable to determine if a linear system game has a
    perfect strategy in $C_q$. 
\end{cor}

\subsection{Acknowledgements}

I thank Jason Crann, Richard Cleve, Tobias Fritz, Li Liu, Martino Lupini,
Narutaka Ozawa, Vern Paulsen, Mark Sapir, Jamie Sikora, and Thomas Vidick for helpful
comments and conversations. 

\section{Group theory preliminaries}

\subsection{Group presentations}

Given a set $S$, let $\mcF(S)$ denote the free group generated by $S$. If $H$
is a group, then homomorphisms $\mcF(S) \arr H$ can be identified with
functions $S \arr H$, and we use these two types of objects interchangeably.
If $R$ is a subset of $\mcF(S)$, then the quotient of $\mcF(S)$ by the normal
subgroup generated by $R$ is denoted by $\langle S : R \rangle$.  If $G =
\langle S : R \rangle$ and $R' \subset \mcF(S \cup S')$,
then we write $\langle G, S' : R' \rangle$ to mean $\langle S \cup S' : R \cup
R' \rangle$. 

A group $G$ is said to be \emph{finitely presentable} if $G = \langle S : R
\rangle$ for some finite sets $S$ and $R$. A \emph{finitely presented group} is
a tuple $(G,S,R)$, where $G = \langle S : R \rangle$. In other words, a
finitely presented group is a finitely presentable group along with a choice of
finite presentation. 

\subsection{Approximate representations}

Let $\norm{\cdot}$ be the normalized Hilbert-Schmidt norm, i.e. if $T$ is an
endomorphism of a finite-dimensional Hilbert space $H$, then $\norm{T} =
\sqrt{\tr(T^* T)} / \sqrt{\dim H}$. 
\begin{defn}\label{D:approx}
    Let $G = \langle S : R \rangle$ be a finitely presented group. A
    \emph{finite-dimensional $\eps$-approximate representation} (or
    \emph{$\eps$-representation} for short) is a homomorphism $\phi : \mcF(S)
    \arr \mcU(H)$ from $\mcF(S)$ to the unitary group $\mcU(H)$ of some
    finite-dimensional Hilbert space $H$, such that
    \begin{equation*}
        \norm{\phi(r) - \Id} \leq \eps
    \end{equation*}
    for all $r \in R$. 
\end{defn}
Note that the normalized Hilbert-Schmidt norm is invariant under conjugation by
unitaries, so the set of $\eps$-representations is independent of the cyclic
order of the relations $r \in R$. That means that, for instance, we can write
the relation $x = y$ without worrying about whether we mean $xy^{-1} = e$ or
$y^{-1} x = e$. 

There are several different notions of approximate representations in the
literature. The notion we are using comes from the study of stable relations of
$C^*$-algebras (see, for instance, Section 4.1 of \cite{Lo97}). For the
purposes of this paper, we could also use the closely related notion of
approximate homomorphisms as in \cite[Section II]{CL15}. However, Definition
\ref{D:approx} is very convenient for working with examples, as we frequently
do in this paper. The main disadvantage of this definition is that it depends
on the choice of presentation. We can work around this using the following
easy lemma.
\begin{lemma}\label{L:pullback}
    Let $\psi : G \arr H$ be a homomorphism, where $G = \langle S : R\rangle$
    and $H = \langle S' : R' \rangle$ are finitely presented groups. If
    $\Psi : \mcF(S) \arr \mcF(S')$ is a lift of $\psi$, then there
    is a constant $C > 0$ such that if $\phi$ is an $\eps$-representation of
    $H$, then $\phi \circ \Psi$ is a $C \eps$-representation of $G$.
\end{lemma}
We record two other simple lemmas for later use.
\begin{lemma}\label{L:smallchanges}
    Let $G = \langle S : R \rangle$, and let $M$ be the length of the longest
    relation in $R$. If $\phi$ is an $\eps$-representation of $G$, and $\psi$
    is an approximate representation of $G$ with
    \begin{equation*}
        \norm{\psi(x) - \phi(x) } \leq \delta
    \end{equation*}
    for all $x \in S$, then $\psi$ is an $(M\delta+\eps)$-representation.
\end{lemma}
Given approximate representations $\phi : \mcF(S) \arr \mcU(H)$ and $\psi :
\mcF(S) \arr \mcU(H')$ of $G = \langle S : R \rangle$, we can form new
approximate representations $\phi \oplus \psi : \mcF(S) \arr \mcU(H\oplus H')$
and $\phi \otimes \psi : \mcF(S) \arr \mcU(H \otimes H')$. 
\begin{lemma}\label{L:sumandproduct}
    Suppose $\phi$ and $\psi$ are $\eps$- and $\eps'$-representations of $G$
    respectively. Then $\phi \oplus \psi$ is a $\max(\eps,\eps')$-representation,
    and $\phi \otimes \psi$ is an $(\eps + \eps')$-representation.
\end{lemma}

A group $G$ is said to be residually finite-dimensional if every non-trivial
element of $G$ is non-trivial in some finite-dimensional representation. More
generally, the set of elements which are trivial in finite-dimensional
representations forms a normal subgroup of $G$. We let $G^{fin}$ denote the
quotient of $G$ by this normal subgroup (alternatively, $G^{fin}$ is the image
of $G$ in its profinite completion). Any homomorphism $\phi : G \arr H$
descends to a homomorphism $G^{fin} \arr H^{fin}$. 
\begin{defn}
    A homomorphism $\phi : G \arr H$ is a \emph{$fin$-embedding} if
    the induced map $G^{fin} \arr H^{fin}$ is injective, and a \emph{$fin^*$-embedding}
    if $\phi$ is both injective and a $fin$-embedding.
\end{defn}
Equivalently, $\phi$ is a $fin$-embedding if $\phi(g)$ is non-trivial in
finite-dimensional representations whenever $g \in G$ is non-trivial in
finite-dimensional representations. 

We can similarly look at elements which are non-trivial in approximate
representations:
\begin{defn}\label{D:nontrivial}
    Let $G$ be a finitely presentable group. An element $g \in G$ is
    \emph{non-trivial in (finite-dimensional) approximate representations}
    if there is a finite presentation $G = \langle S : R\rangle$, a
    representative $w \in \mcF(S)$ for $g$, and some constant $\delta > 0$ such
    that, for all $\eps > 0$, there is an $\eps$-representation $\phi$ of $G$
    with $\norm{\phi(w) - \Id} > \delta$.
\end{defn}
Alternatively, if $g \in G = \langle S : R \rangle$, let 
\begin{equation*}
    \ell^{fa}(g) := \lim_{\eps \arr 0^+} \sup_{\phi} \norm{\phi(w) - \Id},
\end{equation*}
where $w$ is a representative for $g$, and the supremum is across
$\eps$-representations $\phi$ of $G$. It is easy to see that the right-hand
side is independent of the choice of representative $w$. By Lemma
\ref{L:pullback}, if $\psi : G \arr H$ is a homomorphism, then $\ell^{fa}(g)
\geq \ell^{fa}(\psi(g))$. Consequently, $\ell^{fa}(g)$ is independent of the
chosen presentation $\langle S : R \rangle$, and $g$ is non-trivial in
approximate representations if and only if $\ell^{fa}(g) > 0$. This makes
it apparent that the choice of presentation $\langle S : R\rangle$ and
representative $w$ in Definition \ref{D:nontrivial} is arbitrary. 

Standard amplification arguments show that the constant $\delta$ in Definition
\ref{D:nontrivial} is also somewhat arbitrary; in fact, $\ell^{fa}(g)$ never
takes values in $(0,\sqrt{2})$. The same amplification arguments can be used to
show that a finitely-presented group $G$ is \emph{hyperlinear} if and only if
every non-trivial element of $G$ is non-trivial in approximate representations,
and this can be used as the definition of hyperlinearity for finitely-presented
groups. We refer Section II.2 of \cite{CL15} for the standard definition of
hyperlinearity, along with the amplification arguments needed to prove the
equivalence. 

Clearly $\ell^{fa}(g) \geq 0$ for all $g \in G$, and it is easy to see that
$\ell^{fa}(gh) \leq \ell^{fa}(g) + \ell^{fa}(h)$ and $\ell^{fa}(hgh^{-1}) =
\ell^{fa}(g)$ for all $g,h \in G$. Thus the set of elements of $G$ which are
trivial in approximate representations (i.e. for which $\ell^{fa}(g) = 0$)
forms a normal subgroup of $G$. Let $G^{fa}$ be the quotient of $G$ by this
normal subgroup. Because $\ell^{fa}$ is decreasing via homomorphisms, any
homomorphism $\phi : G \arr H$ between finitely presentable groups descends to
a homomorphism $G^{fa} \arr H^{fa}$. 
\begin{defn}
    A homomorphism $\phi : G \arr H$ is an \emph{$fa$-embedding} if the induced
    map $G^{fa} \arr H^{fa}$ is injective, and an \emph{$fa^*$-embedding} if
    $\phi$ is injective, a $fin$-embedding, and an $fa$-embedding. 
\end{defn}
Equivalently, $\phi$ is an $fa$-embedding if $\phi(g)$ is non-trivial in
approximate representations whenever $g \in G$ is non-trivial in approximate
representations. 

If $\phi$ and $\psi$ are approximate representations, then we say that $\phi$
is a \emph{direct summand} of $\psi$ if $\psi = \phi \oplus \phi'$ for some other
approximate representation $\phi'$. We use the following simple trick to
construct $fa^*$-embeddings. 
\begin{lemma}\label{L:trick}
    Let $G = \langle S : R \rangle$ and $H = \langle S' : R' \rangle$ be
    two finitely presented groups, and let $\Psi : \mcF(S) \arr
    \mcF(S')$ be a lift of a homomorphism $\psi : G \arr H$. 
    \begin{enumerate}[(a)]
        \item Suppose that for every representation (resp. finite-dimensional
            representation) $\phi$ of $G$, there is a representation (resp.
            finite-dimensional representation) $\gamma$ of $H$ such that $\phi$
            is a direct summand of $\gamma \circ \psi$. Then $\psi$ is injective
            (resp. a $fin$-embedding).
        \item Suppose that there is an integer $N>0$ and a real number $C>0$ such that
            for every $d$-dimensional $\eps$-representation $\phi$ of $G$, where
            $\eps > 0$, there is an $Nd$-dimensional $C \eps$-representation
            $\gamma$ of $H$ such that $\phi$ is a direct summand of $\gamma
            \circ \Psi$. Then $\psi$ is an $fa$-embedding.
    \end{enumerate}
\end{lemma}
\begin{proof}
    Part (a) is clear, so we prove (b).
    Suppose $\phi$ is an $\eps$-representation of $G$, where $\eps > 0$.
    If $\gamma \circ \Psi = \phi \oplus \phi'$, where $\phi$ is $d$-dimensional
    and $\phi'$ is $(N-1)d$-dimensional, then
    \begin{equation*}
        \norm{\gamma(\Psi(w)) - \Id} = \norm{\phi(w)\oplus\phi'(w)
            - \Id} \geq \frac{1}{\sqrt{N}} \norm{\phi(w)-\Id}
    \end{equation*} 
    for all $w \in \mcF(S)$. So $\ell^{fa}(\psi(g)) \geq \ell^{fa}(g)/\sqrt{N}$,
    and $\psi$ is an $fa$-embedding. 
\end{proof}
In our applications it will be possible to check parts (a) and
(b) of Lemma \ref{L:trick} simultaneously, in which case $\psi$ will be an
$fa^*$-embedding.

\subsection{Groups over $\Z_2$}

For convenience, we use the following definition from from \cite{Sl16}: A
\emph{group over $\Z_2$} is a pair $(G,J)$, where $J$ is a central element of
$G$ of order two. Note that $J$ is allowed to be the identity element.
Typically we drop the pair notation, and just use the symbol $J$ (or $J_G$
where necessary) to refer to the special element of a group $G$ over $\Z_2$, in
the same way that we use $e$ to refer to the identity element. If $G$ and $H$
are groups over $\Z_2$, then a \emph{morphism $G \arr H$ over $\Z_2$} is a
group homomorphism $G \arr H$ which sends $J_G \mapsto J_H$. 

If a group $G$ over $\Z_2$ is finitely presentable, then it has a finite
presentation $\langle S : R \rangle$ where $J \in S$, and $R$ includes the
relations $J^2 = e$ and $[J,s]=e$ for every $s \in S \setminus \{J\}$. We use
presentations of this form often enough that it is helpful to have some
notation for them.  Suppose that $S_0$ is a set of indeterminates, and $R_0
\subset \mcF(S_0 \cup \{J\})$.  Then we set
\begin{equation*}
    \langle S_0 : R_0 \rangle_{\Z_2} := \left\langle S_0 \cup \{J\} : R_0 \cup
        \left\{[J,s] = e \; : \; s\in S_0\right\} \cup \{J^2 = e\} \right\rangle,
\end{equation*}
and call $\langle S_0 : R_0 \rangle_{\Z_2}$ a \emph{presentation over $\Z_2$}.
As with ordinary presentations, if $G = \langle S : R \rangle$ or $\langle S : R
\rangle_{\Z_2}$, then $\langle G, S' : R' \rangle_{\Z_2} := \langle S \cup S' :
R \cup R' \rangle_{\Z_2}$. 

\section{Linear system games and solution groups}

Let $Ax=b$ be an $m \times n$ linear system over $\Z_2$. To the system $Ax=b$,
we can associate a non-local game, called a \emph{linear system game}, as
follows. For each $1 \leq i \leq m$, let $V_i = \{j : A_{ij} \neq 0\}$ be the
set of indices of variables appearing in the $i$th equation. Let $S_i \subset
\Z_{2}^{V_i}$ be the set of assignments to variables $x_j$, $j \in V_j$
satisfying the $i$th equation, i.e.  $\underline{a} \in \Z_{2}^{V_i}$ belongs
to $S_i$ if and only if $\sum_{j \in V_j} a_j = b_i$. Then Alice receives an
equation as input, represented by an integer $1 \leq i \leq m$, and must output
an element $\underline{a} \in S_i$. Bob receives a variable, represented by an
integer $1 \leq j \leq n$, and must output an assignment $b$ for $x_j$.  The
players win if either $j \not\in V_i$, or $j \in V_i$ and $a_j = b$, i.e.
Alice's and Bob's outputs are consistent.

A \emph{quantum strategy (presented in terms of measurements)} for a linear
system game consists of 
\begin{enumerate}[(1)]
    \item a pair of Hilbert spaces $H_A$ and $H_B$, 
    \item a projective measurement $\{N^j_{b}\}_{b\in \Z_2}$ on $H_B$ for every integer $1
        \leq j \leq n$,
    \item a projective measurement $\{M^i_{\underline{a}}\}_{\underline{a} \in
        S_i}$ on $H_A$ for every integer $1 \leq i \leq m$, and
    \item a quantum state $\ket{\psi} \in H_A \otimes H_B$.
\end{enumerate}
The strategy is \emph{finite-dimensional} if $H_A$ and $H_B$ are finite-dimensional.
The associated quantum correlation matrix $\{p(\underline{a},b|i,j)\}$ is
defined by 
\begin{equation*}
    p(\underline{a},b|i,j) = \bra{\psi} M^{i}_{\underline{a}} \otimes N^j_{b} \ket{\psi}, \quad
        1 \leq i \leq m,\ 1 \leq j \leq n,\ \underline{a} \in S_i,\ b \in \Z_2.
\end{equation*}
As in the introduction, we also use the term \emph{strategy} to refer to the
correlation matrix $\{p(\underline{a},b|i,j)\}$. If $j \in V_i$, then the
probability that Alice and Bob win on inputs $i$ and $j$ is 
\begin{equation*}
    p_{ij} := \sum_{\underline{a},b : a_j = b} p(\underline{a},b|i,j).
\end{equation*}
A strategy is \emph{perfect} if and only if $p_{ij} = 1$ for all $1 \leq i \leq
m$ and $j \in V_i$.

For linear system games, it is often convenient to work with strategies
presented in terms of $\pm 1$-valued observables---self-adjoint operators
which square to the identity---rather than measurement operators. 
A \emph{quantum strategy (presented in terms of observables)} consists of 
\begin{enumerate}[(a)]
    \item a pair of Hilbert spaces $H_A$ and $H_B$; 
    \item a collection of self-adjoint operators $X_j$, $1 \leq j \leq n$, on $H_B$ such
        that $X_j^2 = \Id$ for every $1 \leq j \leq n$; 
    \item a collection of self-adjoint operators $Y_{ij}$, $1 \leq i \leq m$, $j \in V_i$ on $H_A$ 
        such that
        \begin{enumerate}[(i)]
            \item $Y_{ij}^2 = \Id$ for every $1 \leq i \leq m$ and $j \in V_i$, 
            \item $\prod_{j \in V_i} Y_{ij} = (-\Id)^{b_i}$ for every $1 \leq i \leq m$, and
            \item $Y_{ij} Y_{il} = Y_{il} Y_{ij}$ for every $1 \leq i \leq m$ and $j,l \in V_i$; 
        \end{enumerate}
        and
    \item a quantum state $\ket{\psi} \in H_A \otimes H_B$. 
\end{enumerate}
Given a quantum strategy presented in terms of measurements, we can get a
quantum strategy presented in terms of observables by setting $X_j = N^j_0 -
N^j_1$ for every $1 \leq j \leq n$, and 
\begin{equation*}
    Y_{ij} = \sum_{\underline{a} \in S_i} (-1)^{a_j} M^{i}_{\underline{a}}
\end{equation*}
for $1 \leq i \leq m$ and $j \in V_i$. Conversely, given a
quantum strategy in terms of observables, we can recover the measurement
presentation using the spectral decomposition of the observables. So the two
notions of strategy are equivalent. Note that if $j \in V_i$, then
\begin{equation}\label{E:winningbias}
\begin{split}
    \bra{\psi} Y_{ij} \otimes X_j \ket{\psi} & = \bra{\psi} \sum_{\underline{a} \in S_i}
        (-1)^{a_j} M^i_{\underline{a}} \otimes \sum_{b \in \Z_2} (-1)^b N^j_b \ket{\psi} \\
        & = \sum_{\underline{a} \in S_i, b \in \Z_2} (-1)^{a_j + b} p(\underline{a},b|i,j) \\
        & = 2 \left[ \sum_{\underline{a},b : a_j = b} p(\underline{a},b|i,j) \right] - 1 = 2p_{ij} - 1,
\end{split}
\end{equation}
where $p_{ij}$ is, again, the probability that Alice
and Bob win on inputs $i$ and $j$. The quantity $2p_{ij}-1$ is called the
\emph{winning bias} on inputs $i$ and $j$.

To every linear system, we can also associate a finitely presented group over
$\Z_2$, as follows. 
\begin{defn}\label{D:solutiongroup}
    Let $Ax=b$ be an $m \times n$ linear system. The \emph{solution group} 
    of this system is the group
    \begin{align*}
        \Gamma(A,b) : = \Big\langle x_1,\ldots,x_n :\; & x_n^2 = e \text{ for all } 1 \leq j \leq n, \\
                                               & \prod_{j=1}^n x_j^{A_{ij}} = J^{b_i} \text{ for all } 
                                                    1 \leq i \leq m, \text{ and } \\
                                               & x_j x_k = x_k x_j \text{ if } j,k \in V_i
                                            \text{ for some } 1 \leq i \leq m \; \Big\rangle_{\Z_2}
    \end{align*}
    We say that a group over $\Z_2$ is a \emph{solution group} if it has a
    presentation over $\Z_2$ of this form. 
\end{defn}
Solution groups and linear system games are related as follows.
\begin{thm}[\cite{CM14}, see also \cite{CLS16}]\label{T:solutiongroup} 
    Let $\mcG$ be the linear system game associated to a system $Ax=b$.
    Then the following are equivalent:
    \begin{enumerate}[(a)]
        \item $\mcG$ has a perfect strategy in $C_{qs}$.
        \item $\mcG$ has a perfect strategy in $C_{q}$.
        \item $J_{\Gamma}$ is non-trivial in some finite-dimensional
            representation of $\Gamma = \Gamma(A,b)$. 
    \end{enumerate}
\end{thm}

Although we haven't defined the set of commuting-operator correlations
$C_{qc}$, we can work with $C_{qc}$ through the following result.
\begin{thm}[\cite{CLS16}]\label{T:solutiongroup2}
    The linear system game associated to a system $Ax=b$ has a perfect
    strategy in $C_{qc}$ if and only if $J_{\Gamma}$ is non-trivial
    in $\Gamma = \Gamma(A,b)$. 
\end{thm}

The main point of this section is to prove an analog of one direction of
Theorem \ref{T:solutiongroup} for approximate representations.
\begin{prop}\label{P:reptostrat}
    Let $\Gamma = \Gamma(A,b)$ be a solution group. If $J_{\Gamma}$ is
    non-trivial in finite-dimensional approximate representations of $\Gamma$
    then the linear system game associated to $Ax=b$ has a perfect strategy in
    $C_{qa}$. 
\end{prop}

The proof of Proposition \ref{P:reptostrat} is a straightforward application of
a number of easy stability lemmas. We start by pinning down what we want to
prove.  
\begin{lemma}\label{L:perfect}
    The linear system game associated to $Ax=b$ has a perfect strategy
    in $C_{qa}$ if and only if, for all $\eps > 0$, there is a
    finite-dimensional quantum strategy (presented in terms of observables)
    $\{Y_{ij}\}$, $X_j$, $\ket{\psi}$ such that
    \begin{equation*}
        \bra{\psi} Y_{ij} \otimes X_j \ket{\psi} \geq 1 - \eps \text{ for all }
            1 \leq i \leq n, j \in V_i.
    \end{equation*}
\end{lemma}
\begin{proof}
    Since $C_{qa}$ is the closure of $C_q$, the linear system game associated
    to $Ax=b$ has a perfect strategy in $C_{qa}$ if and only if, for every
    $\eps > 0$, there is a finite-dimensional quantum strategy such that the
    winning probability $p_{ij} \geq 1 - \eps/2$ for every $1 \leq i \leq m$
    and $j \in V_i$. But $p_{ij} \geq 1-\eps/2$ if and only if the winning
    bias $2p_{ij} - 1 \geq 1 -\eps$, so the lemma follows from equation
    \eqref{E:winningbias}. 
\end{proof}

Next, we come to the stability lemmas, which will allow us to turn approximate
representations of the solution group $\Gamma$ into quantum strategies. The
following lemmas are all likely well-known to experts (see, for instance,
\cite{Gl10,FK10}); we include the proofs for completeness.
\begin{lemma}\label{L:stability1}
    For any diagonal matrix $X$, there is a diagonal matrix $D$ with $D^2 =
    \Id$ and 
    \begin{equation*}
        \norm{D-X} \leq \left(1 + \frac{1}{\sqrt{2}}\right) \norm{X^2 - \Id}.
    \end{equation*}
\end{lemma}
\begin{proof}
    Suppose $X$ is a $d \times d$ matrix, and let $D_{ii} = \sgn \Re X_{ii}$ for
    all $1  \leq i \leq d$, where $\sgn x = 1$ if $x \geq 0$ and $-1$ if $x <
    0$. To show that the desired inequality holds, consider a complex number $\alpha
    = a + bi$. Then 
    \begin{align*}
        |\alpha^2 - 1|^2 & = |a^2-b^2-1 + 2abi|^2 = \left[\left(a^2-1\right)-b^2\right]^2 + 4 a^2 b^2 \\
            & = (a^2-1)^2 + 2 b^2 + 2 a^2 b^2 + b^4.
    \end{align*}
    In particular, this implies that $|\alpha^2-1|^2$ is greater than or equal to
    $(a^2-1)^2$ and $2 b^2$. Consequently,
    \begin{align*}
        \norm{ \left(\Re X\right)^2 - \Id } 
             = \sqrt{ \frac{1}{d} \sum_j \left[\left(\Re X_{jj}\right)^2 - 1\right]^2 } 
             \leq \sqrt{\frac{1}{d} \sum_j | X_{jj}^2 - 1 |^2 } = \norm{ X^2 - \Id},
    \end{align*}
    and 
    \begin{align*}
        \norm{\Re X - X} = \norm{\Im X} = \sqrt{ \frac{1}{d} \sum_j |\Im X_{jj} |^2 }
            \leq \sqrt{\frac{1}{2d} \sum_j | X_{jj}^2 - 1|^2 } = \frac{1}{\sqrt{2}} \norm{X^2 - \Id}.
    \end{align*}
    By considering the cases $a \geq 0$ and $a < 0$ separately, we see that
    \begin{equation*}
        |a^2 - 1| = |1 +a| |1-a| = (1 + |a|)|\sgn a - a| \geq |\sgn a - a|
    \end{equation*}
    for all $a \in \R$. 
    Thus, as above, $\norm{D - \Re X} \leq \norm{(\Re X)^2-\Id}$, and the lemma follows. 
\end{proof}

\begin{lemma}\label{L:stability2}
    Suppose $X_1,\ldots,X_n$ are commuting unitary matrices, with $X_i^2 = \Id$
    for all $1  \leq i \leq n$, and $Y$ is a unitary matrix such that $Y^2 =
    \Id$ and $Y$ commutes with $X_i$ for all $1 \leq i \leq n-1$. Then there is a
    unitary matrix $Z$ such that $Z^2 = \Id$, $Z$ commutes
    with $X_i$ for all $1 \leq i \leq n$, and
    \begin{equation*}
        \norm{Z-Y} \leq \left(1 + \frac{1}{2\sqrt{2}} \right) \norm{X_n Y - Y X_n}.
    \end{equation*}
\end{lemma}
\begin{proof}
    Let $Z_0 = \frac{1}{2} (Y + X_n Y X_n)$. Clearly $Z_0$ commutes with $X_i$ for
    all $1 \leq i \leq  n-1$. Since $X_n^2 = \Id$, we also have that $X_n Z_0 =
    \frac{1}{2} (X_n Y + Y X_n) = Z_0 X_n$. Since $Y^2 = \Id = (X_n Y X_n)^2$ as
    well, we have that
    \begin{align*}
        \norm{Z_0^2 - \Id} & = \frac{1}{4} \norm{Y X_n Y X_n + X_n Y X_n Y - 2 \Id}   \\
            & \leq \frac{1}{4} \norm{Y X_n Y X_n - \Id} + \frac{1}{4} \norm{X_n Y X_n Y - \Id}
             = \frac{1}{2} \norm{X_n Y - Y X_n}.
    \end{align*}
    Since $X_n$ and $Y$ are self-adjoint, $Z_0$ is self-adjoint, so we can simultaneously
    diagonalize $X_1,\ldots,X_n$ and $Z_0$. Hence by Lemma \ref{L:stability1},
    there is a matrix $Z$ such that $Z^2 = \Id$, $Z$ commutes
    with $X_i$ for all $1 \leq i \leq n$, and
    \begin{align*}
        \norm{Z - Z_0} \leq \left(1 + \frac{1}{\sqrt{2}}\right) \norm{Z_0^2 - \Id}
            \leq \left(\frac{1}{2} + \frac{1}{2\sqrt{2}}\right) \norm{X_n Y - YX_n}.
    \end{align*}
    Finally,
    \begin{equation*}
        \norm{Y - Z_0} = \frac{1}{2} \norm{Y - X_n Y X_n} = \frac{1}{2} \norm{X_n Y - Y X_n},
    \end{equation*}
    so the lemma follows.  
\end{proof}

\begin{lemma}\label{L:stability3}
    Consider $\Z_2^k$ as a finitely-presented group with presentation 
    \begin{equation*}
        \langle x_1,\ldots,x_k : x_i^2 = e, [x_i,x_j] = e \text{ for all } i \neq j \rangle.
    \end{equation*}
    Then there is a constant $C>0$, depending on $k$, such that if $\phi$ is an
    $\eps$-representation of $\Z_2^k$ on a Hilbert space $H$, then there is a
    representation $\psi$ of $\Z_2^k$ on $H$ with 
    \begin{equation*}
        \norm{ \psi(x_i) - \phi(x_i) } \leq C \eps
    \end{equation*}
    for all $1 \leq i \leq k$.
\end{lemma}
\begin{proof}
    Suppose $\psi$ is an $\eps$-representation of $\Z_2^k$ such that the following
    properties hold for some $1 \leq l \leq k-1$: 
    \begin{enumerate}[(a)]
        \item $\psi(x_i)^2 = \Id$ for all $1 \leq i \leq k$, and
        \item $\psi(x_i)$ commutes with $\psi(x_j)$ for all $1 \leq i \leq l-1$ and $1 \leq j \leq k$.
    \end{enumerate}
    In particular, property (b) requires that $\psi(x_1),\ldots,\psi(x_l)$
    pairwise commute. Then by Lemma \ref{L:stability2}, for each $l < j \leq k$
    there is a unitary matrix $X_j$ such that $X_j^2 = \Id$, $X_j$ commutes
    with $\psi(x_i)$ for all $1 \leq i \leq l$, and 
    \begin{equation*}
        \norm{X_j - \psi(x_j)} \leq C_0 \norm{\psi(x_l)\psi(x_j) - \psi(x_j)\psi(x_l)} \leq C_0 \eps,
    \end{equation*}
    where $C_0 = 1 + \frac{1}{2\sqrt{2}}$. Define an approximate representation
    $\psi'$ of $G$ by $\psi'(x_i) = \psi(x_i)$ if $i \leq l$ and $\psi'(x_i) =
    X_i$ if $i > l$. Then $\psi'(x_i)^2 = \Id$ for all $1  \leq i \leq k$, and
    $\psi'(x_i)$ commutes with $\psi'(x_j)$ for all $1 \leq i \leq l$ and $1
    \leq j \leq k$. In other words, $\psi'$ satisfies properties (a) and (b)
    with $l$ replaced by $l+1$. Finally, $\norm{\psi'(x_i) - \psi(x_i)} \leq
    C_0 \eps$ for all $1 \leq i \leq k$, so $\psi'$ is a
    $(4C_0+1)\eps$-representation by Lemma \ref{L:smallchanges}. 

    Now suppose that $\phi$ is any $\eps$-representation of $\Z_2^k$. By Lemma
    \ref{L:stability1}, there is an approximate representation $\psi_1$ of
    $\Z_2^k$ with $\psi_1(x_i)^2 = \Id$ and $\norm{\psi_1(x_i) - \phi(x_i)}
    \leq C_1 \eps$ for all $1 \leq i \leq k$, where $C_1 = (1 +
    \frac{1}{\sqrt{2}})$. By Lemma \ref{L:smallchanges}, $\psi_1$ is a
    $(4C_1+1)\eps$-representation. Clearly, $\psi_1$ satisfies conditions (a)
    and (b) with $l=1$.  Using the argument in the previous paragraph, we can
    then iteratively define approximate representations
    $\psi_2,\ldots,\psi_{k-1}$, where $\psi_j$ satisfies conditions (a) and (b)
    with $l=j$ for all $1 \leq j \leq k-1$. Let $\eps_l = (4C_0+1)^{l-1} (4C_1+1)
    \eps$, so $\psi_1$ is an $\eps_1$-representation. It is not hard to check
    that $\psi_l$ is an $\eps_l$-representation, and furthermore that 
    \begin{equation*}
        \norm{\psi_l(x_i) - \psi_1(x_i)} \leq \frac{1}{4} \left( (4C_0+1)^{l-1} - 1\right) \eps_1
            = \frac{1}{4} \left( (4C_0+1)^{l-1} - 1 \right) (4C_1+1) \eps
    \end{equation*}
    for all $1 \leq i \leq k$. Since $\psi_{k-1}$ is an exact representation, 
    we can take 
    \begin{equation*} 
        C = \frac{1}{4} \left( (4C_0+1)^{k-2} - 1 \right) (4C_1+1) + C_1.
    \end{equation*}
\end{proof}
        
\begin{lemma}\label{L:stability4}
    Suppose $G = \langle S_0 : R_0 \rangle_{\Z_2}$, where $R_0$ includes
    the relations $s^2 = e$ for all $s \in S_0$. If $J_G$ is non-trivial 
    in finite-dimensional approximate representations of $G$, then for every
    $\eps > 0$ there is an $\eps$-representation $\phi$ of $G$ such that
    $\phi(J)=-\Id$, and $\phi(s)^2 = \Id$ for all $s \in S_0$.
\end{lemma}
\begin{proof}
    Suppose $A$ is an $m \times n$ matrix, and let $S = S_0 \cup \{J\}$.
    If $J$ is non-trivial in approximate representations, then there is a
    $\delta > 0$ such that for all $\eps > 0$, there is an
    $\eps$-representation $\phi$ with $\norm{\phi(J)-\Id} > \delta$. 

    By Lemmas \ref{L:smallchanges}, \ref{L:stability1}, and \ref{L:stability2},
    there are constants $C,C' >0$ such that if $\phi$ is an $\eps$-representation,
    then there is a $C'\eps$-representation $\psi$ such that
    \begin{enumerate}[(1)]
        \item $\psi(x)^2 = \Id$ for all $x \in S$, 
        \item $\psi(s)$ and $\psi(J)$ commute for all $s \in S_0$, and
        \item $\norm{\psi(J) - \phi(J)} \leq C \eps$.
    \end{enumerate}
    (We can take $C = (1+\frac{1}{\sqrt{2}})$, while $C'$ will depend on the 
    length of the longest defining relation of $G$.)
    If $\norm{\phi(J) - \Id} > \delta$, and $\eps < \delta / (2C)$, then 
    \begin{equation*}
        \delta < \norm{\phi(J) - \Id} \leq \norm{\phi(J) - \psi(J)} + \norm{\psi(J)-\Id} \leq \frac{\delta}{2}
            + \norm{\psi(J)-\Id},
    \end{equation*}
    so $\norm{\psi(J)-\Id} \geq \frac{\delta}{2}$. Thus we conclude that
    for all $\eps > 0$, there is an $\eps$-representation $\psi$ satisfying
    conditions (1) and (2), and with $\norm{\psi(J) - \Id} > \frac{\delta}{2}$. 

    Suppose $\psi$ is an $\eps$-representation satisfying conditions (1) and (2),
    and with $\norm{\psi(J) - \Id} > \frac{\delta}{2}$. Choose a basis
    with $\psi(J) = \Id_{d_0} \oplus (-\Id_{d_1})$. Since $\psi(s)$ commutes
    with $\psi(J)$ for all $s \in S_0$, we must have $\psi = \psi_0
    \oplus \psi_1$, where $\psi_a$ is an approximate representation of
    dimension $d_a$, and $\psi_a(J) = (-\Id)^a$, $a = 0,1$. Since $\psi(s)^2=\Id$,
    we also have $\psi_a(s)^2=\Id$ for all $s \in S_0$, $a=0,1$. To finish
    the proof, we just need to show that $\psi_1$ is a $C''\eps$-representation for
    some constant $C''$ independent of $\psi$.  If $w \in \mcF(S)$, then
    \begin{equation*}
        \norm{\psi(w) - \Id}^2 = \frac{d_0}{d_0+d_1} \norm{\psi_0(w)-\Id}^2
            + \frac{d_1}{d_0+d_1} \norm{\psi_1(w)-\Id}^2.
    \end{equation*}
    If $w = J$, then $\norm{\psi_0(w) - \Id}=0$ and $\norm{\psi_1(w)-\Id} =
    \norm{-2\Id} = 4$, so we conclude that
    \begin{equation*}
        \frac{\delta^2}{4} < \norm{\psi(J) - \Id}^2 = \frac{4 d_1}{d_0+d_1},
    \end{equation*}
    so $d_1/(d_0+d_1) > \delta^2 / 16$. On the other hand, if $w = r$ is 
    one of the defining relations of $G$, then
    \begin{equation*}
        \eps^2 \geq \norm{\psi(r) - \Id}^2 \geq \frac{d_1}{d_0+d_1} \norm{\psi_1(r)- \Id}^2 
            > \frac{\delta^2}{16} \norm{\psi_1(r) - \Id}^2.
    \end{equation*}
    Thus $\psi_1$ is a $4\eps/\delta$-representation with $\psi_1(J)=-\Id$ and
    $\psi_1(s)^2 = \Id$ for all $s \in S_0$. Since $\delta$ is a
    constant, the lemma follows. 
\end{proof}

\begin{proof}[Proof of Proposition \ref{P:reptostrat}]
    Suppose $J$ is non-trivial in finite-dimensional approximate
    representations of $\Gamma$. Given $\eps > 0$, let $\phi$ be an
    $\eps$-representation of $\Gamma$ with $\phi(J) = -\Id$ and $\phi(x_j)^2 =
    \Id$ for all $1 \leq j \leq n$, as in Lemma \ref{L:stability4}. Suppose
    $\phi$ has dimension $d$, and let $\ket{v}$ be the maximally entangled
    state on $\C^d \otimes \C^d$. For each $1 \leq j \leq n$, set $X_j =
    \phi(x_j)$. For each $1 \leq i \leq m$, let $j_i$ be the maximal element
    of $V_i$, and set $W_i := V_i \setminus \{j_i\}$. The restriction of $\phi$
    to the subgroup $\langle x_j : j \in W_i \rangle$ is an $\eps$-representation
    of $\Z_2^{W_i}$, and by Lemma \ref{L:stability3}, there is a representation
    $\psi_i$ of $\Z_2^{W_i}$ with $\norm{\psi_i(x_j) - \phi(x_j)} \leq O(\eps)$.\footnote{For this proof, we use the notation $O(\eps)$ to hide constants
    which are independent of $\eps$, $\phi$, and so on. The constants can still
    depend on the linear system $Ax=b$, however.}
    Set $Y_{ij} := \psi_i(x_j)^T$ (the transpose of $\psi_i(x_j)$ in a Schmidt
    basis for $\ket{v}$) for all $j \in W_i$, and set $Y_{i{j_i}} := (-1)^{b_i}
    \prod_{j \in W_i} Y_{ij}$. 

    Suppose $j \in W_i$ for some $1 \leq i \leq m$. Since $Y_{ij}$ and $X_j$
    are self-adjoint, we have that
    \begin{equation*}
        2 - \frac{2}{d} \tr(Y_{ij}^T X_j) = \norm{Y_{ij}^T - X_j}^2
            = \norm{\psi(x_j) - \phi(x_j)}^2 \leq O(\eps^2), 
    \end{equation*}
    so $\frac{1}{d} \tr(Y_{ij}^T X_j) \geq 1 - O(\eps^2)$. For the remaining
    variable in $V_i$, we have that
    \begin{align*}
        \norm{Y_{ij_i}^T - X_{j_i}} & = \norm{(-1)^{b_i} \prod_{j \in W_i} \psi_i(x_j) - \phi(x_{j_i})} \\
            & \leq \norm{(-1)^{b_i} \prod_{j \in W_i} \phi(x_{j}) - \phi(x_{j_i})} + |W_i|\eps \\
            & = \norm{(-1)^{b_i} \prod_{j \in V_i} \phi(x_j) - \Id} + |W_i|\eps \leq O(\eps),
    \end{align*}
    where the last equality uses the fact that $\phi(x_{j_i})^2 = \Id$.
    Because the $Y_{ij}$'s commute for all $j \in W_i$, $Y_{ij_i}$ is also
    self-adjoint, so once again we conclude that
    \begin{equation*}
        2 - \frac{2}{d} \tr(Y_{ij_i}^T X_{j_i}) = \norm{Y_{i{j_i}}^T - X_{j_i}}^2
            \leq O(\eps^2)
    \end{equation*}
    or in other words that $\frac{1}{d} \tr(Y_{ij_i}^T X_j) \geq 1 - O(\eps^2)$.  

    Now clearly $\{Y_{ij}\}$, $\{X_j\}$, $\ket{v}$ is a strategy for the 
    linear system game associated to $Ax=b$. If $A$ and $B$ are any two $d
    \times d$ matrices, it follows from the definition of maximally entangled
    states that
    \begin{equation*}
        \bra{v} A \otimes B \ket{v} = \frac{1}{d} \tr(A^T B).
    \end{equation*}
    We conclude that $\bra{v} Y_{ij} \otimes X_j \ket{v} = \frac{1}{d} \tr(Y_{ij}^T X_j)
    \geq 1 - O(\eps^2)$ for all $j \in V_i$, $1 \leq i \leq m$. The proposition follows
    from Lemma \ref{L:perfect}.
\end{proof}

\section{Linear-plus-conjugacy groups}

The goal of the next two sections is to show that there is a solution group
$\Gamma$ such that $J_{\Gamma}$ is trivial in finite-dimensional
representations, but non-trivial in approximate representations. In this
section, we start by showing that it suffices to construct more general types
of group with these properties. 

Given an $m \times n$ linear system $Ax=b$, we once again let $V_i = V_i(A) : =
\{1  \leq j \leq n : A_{ij} \neq 0\}$. 
\begin{defn}\label{D:linearplusconjugacy}
    Suppose $Ax =b$ is an $m \times n$ linear system over $\Z_2$, and $\mcC
    \subseteq [n] \times [n] \times [n]$, where $[n] = \{1,\ldots,n\}$. Let
    \begin{equation*}
        \Gamma(A,b,\mcC) : = \Big\langle \Gamma(A,b) :\; x_i x_j x_i = x_k \text{ for all } (i,j,k) \in \mcC
                                                \Big\rangle_{\Z_2}.
    \end{equation*}
    Lacking a better term, we say that a group over $\Z_2$ is a
    \emph{linear-plus-conjugacy group} if it has a presentation over $\Z_2$ of
    this form. 
\end{defn}
The conjugacy part of the name comes from the fact that since $x_i$ is an
involution, the relation $x_i x_j x_i = x_k$ is equivalent to the relation $x_i
x_j x_i^{-1} = x_k$, so $\Gamma(A,b,\mcC)$ can be thought of as a solution
group with additional conjugacy relations. In the context of
linear-plus-conjugacy and related groups, we use the term \emph{conjugacy
relations} as a convenient shorthand for relations of the form $xyx = z$. We
also use the term \emph{linear relation $x_1 \cdots x_n = e$} to refer to
the set of relations 
\begin{equation*}
    \{x_1 \cdots x_n = e\} \cup \{[x_i,x_j] = e : 1 \leq i \neq j \leq n \}.
\end{equation*}
Finally, observe that there are two ways to make generators $x_i$ and $x_j$
commute in a linear-plus-conjugacy group: we can add a conjugacy relation $x_i
x_j x_i = x_j$, or add an additional generator $x_{n+1}$ and a linear relation
$x_i x_j x_{n+1} = e$. We pick and choose from these two methods based on what
is convenient.

The main point of this section is to prove:
\begin{prop}\label{P:linearplusconjugacy}
    Let $G$ be a linear-plus-conjugacy group. Then there is an $fa^*$-embedding
    $G \arr \Gamma$ over $\Z_2$, where $\Gamma$ is a solution group.  
\end{prop}

We prove Proposition \ref{P:linearplusconjugacy} by first showing that
linear-plus-conjugacy groups can be embedded in linear-plus-conjugacy
groups of a certain form. 
\begin{defn}\label{D:nice}
    A linear-plus-conjugacy group is \emph{nice} if it has a presentation
    of the form $\Gamma(A,b,\mcC)$, where $A$ is an $m \times n$ matrix over
    $\Z_2$, $b \in Z_2^m$, and $\mcC \subseteq [n] \times [n]\times [n]$
    is such that if $(i,j,k) \in \mcC$, then $j,k \in V_l$ for some
    $1 \leq l \leq m$.
\end{defn}
This means that if $x_i x_j x_i = x_k$ is a defining relation of a nice
linear-plus-conjugacy group, then $x_j x_k = x_k x_j$ will also be a defining
relation.
\begin{lemma}\label{L:nice}
    Let $G$ be a linear-plus-conjugacy group. Then there is an $fa^*$-embedding
    $G \arr K$ over $\Z_2$, where $K$ is a nice linear-plus-conjugacy group.
\end{lemma}
\begin{proof}
    Suppose $G = \Gamma(A,b,\mcC)$, where $A$ is an $m\times n$ matrix. Let
    \begin{align*}
        K := \Big\langle \Gamma(A,b)& , w_j, y_j, z_j \text{ for } 1 \leq j \leq n \text{ and } f \; : \\
                & f^2 = e,\; y_j^2 = z_j^2 = w_j^2 = e \text{ for all } 1 \leq j \leq n, \\
                & x_j = y_j z_j = f w_j \text{ and } f y_j f = z_j \text{ for all } 1 \leq j \leq n, \\
                & y_j z_k = z_k y_j \text{ for all } (i,j,k) \in \mcC, \text{ and}\\
                & w_i y_j w_i = z_k \text{ for all } (i,j,k) \in \mcC \Big\rangle_{\Z_2}.
    \end{align*}
    Since the generators are involutions, note that the relations imply that $f
    w_j = w_j f$, $y_j z_j = z_j y_j$, and $f z_j f = y_j$ for all $1 \leq j
    \leq n$. If $(i,j,k) \in \mcC$, then 
    \begin{equation*}
        w_i z_j w_i = w_i f y_j f w_i = f w_i y_j w_i f = f z_k f = y_k, \text{ so}
    \end{equation*}
    \begin{equation*}
        x_i x_j x_i = f w_i y_j z_j f w_i = (f w_i y_j w_i f)(f w_i z_j w_i f)
        = (f z_k f)(f y_k f) = y_k z_k = x_k
    \end{equation*}
    in $K$. Thus there is a homomorphism $\psi : G \arr K$ sending $x_i \mapsto x_i$.

    Suppose $\phi$ is an $\eps$-representation of $G$, where $\eps > 0$.
    Define an approximate representation $\gamma$ of $K$ by 
    \begin{equation*}
        \gamma(x_i) = \begin{pmatrix} \phi(x_i) & 0 \\ 0 & \phi(x_i) \end{pmatrix}, \quad
        \gamma(J) = \begin{pmatrix} \phi(J) & 0 \\ 0 & \phi(J) \end{pmatrix},
    \end{equation*}
    \begin{equation*}
        \gamma(y_i) = \begin{pmatrix} \phi(x_i) & 0 \\ 0 & \Id \end{pmatrix}, \quad
        \gamma(z_i) = \begin{pmatrix} \Id & 0 \\ 0 & \phi(x_i) \end{pmatrix}, 
    \end{equation*}
    \begin{equation*}
        \gamma(w_i) = \begin{pmatrix} 0 & \phi(x_i) \\ \phi(x_i) & 0 \end{pmatrix}, \text{ and }
        \gamma(f) = \begin{pmatrix} 0 & \Id \\ \Id & 0 \end{pmatrix}. 
    \end{equation*}
    It is straightforward to check that $\gamma$ is an $\eps$-representation of
    $K$. If $\Psi$ is the lift of $\psi$ sending $x_i \mapsto x_i$, then
    $\gamma \circ \Psi = \phi \oplus \phi$. When $\phi$ is an exact representation
    of dimension $d$ (possibly infinite), the same construction gives an exact
    representation $\gamma$ of dimension $2d$. By Lemma \ref{L:trick}, $\psi$
    is an $fa^*$-embedding.

    Finally, we observe that $K$ is a nice linear-plus-conjugacy group.
    Indeed, since the relation $x_i = y_i z_i$ forces $y_i$ and $z_i$ to
    commute, this relation is equivalent to the relations
    \begin{equation*}
        x_i y_i z_i = e = [x_i,y_i] = [x_i,z_i] = [y_i,z_i],
    \end{equation*}
    which means that we can make $x_i = y_i z_i$, and similarly $x_i = f w_i$, 
    part of the ``linear'' relations. By adding ancilla variables $g_{jk}$, 
    the commuting relations $y_j z_k = z_k y_j$ can also be replaced with equivalent
    linear relations $g_{jk} y_j z_k = e$. The conjugacy relations
    $f y_j f = z_j$ and $w_i y_j w_i = z_k$ will then satisfy the requirements of
    Definition \ref{D:nice}.
\end{proof}

\begin{proof}[Proof of Proposition \ref{P:linearplusconjugacy}]
    By Lemma \ref{L:nice}, we can assume that $G$ is a nice
    linear-plus-conjugacy group. Let $G = \Gamma(A,b,\mcC)$ be a presentation
    satisfying the conditions of Definition \ref{D:nice}. Augment the linear system
    $Ax=b$ by adding additional variables $y_{Ij}$ for each $I \in \mcC$ and $1
    \leq j \leq 7$, and additional relations 
    \begin{align*}
        & x_i + y_{I1} + y_{I2} = 0, \quad x_j + y_{I2} + y_{I3} = 0, \quad y_{I3} + y_{I4} + y_{I5} = 0 \\
        & x_i + y_{I5} + y_{I6} = 0, \quad x_k + y_{I6} + y_{I7} = 0, \quad y_{I1} + y_{I4} + y_{I7} = 0
    \end{align*}
    for every $I = (i,j,k) \in \mcC$.  Let $\Gamma$ be solution group of this
    augmented linear system, so 
    \begin{equation*}
        \Gamma = \langle\; \Gamma(A,b),\; y_{Ij} \text{ for } I \in \mcC, 1 \leq j \leq 7 \; : \; R\; \rangle_{\Z_2},
    \end{equation*}
    where $R$ consists of the new relations (now written in multiplicative form)
    \begin{align}\label{E:additionalrelations}
        & x_i y_{I1} y_{I2} = x_j y_{I2} y_{I3} = y_{I3} y_{I4} y_{I5} = x_i y_{I5} y_{I6} = x_k y_{I6} y_{I7} = y_{I1} y_{I4} y_{I7} = e
    \end{align}
    for every $I = (i,j,k) \in \mcC$, as well as the corresponding commutation relations.
    In $\Gamma$, we have that
    \begin{align*}
        x_i x_j x_i & = \left(y_{I1} y_{I2}\right) \left(y_{I2} y_{I3}\right) \left(y_{I5} y_{I6}\right) =
            y_{I1} \left(y_{I3} y_{I5} \right) y_{I6} 
        = y_{I1} y_{I4} y_{I6} = y_{I7} y_{I6} = x_k
    \end{align*}
    for every $I = (i,j,k) \in \mcC$.  So once again we get a homomorphism $\psi: G
    \arr \Gamma$ sending $x_i \mapsto x_i$.

    Suppose $\phi$ is an $\eps$-representation of $G$. Define an approximate
    representation $\gamma$ of $\Gamma$ by 
    \begin{align*}
        & \gamma(x_i) = \begin{pmatrix} \phi(x_i) & 0 \\ 0 & \phi(x_i) \end{pmatrix}, 
        & \gamma\left(y_{I1}\right) = \begin{pmatrix} 0 & \phi(x_i) \\ \phi(x_i) & 0 \end{pmatrix}, \\
        & \gamma\left(y_{I2}\right) = \begin{pmatrix} 0 & \Id \\ \Id & 0 \end{pmatrix}, 
        & \gamma\left(y_{I3}\right) = \begin{pmatrix} 0 & \phi(x_j) \\ \phi(x_j) & 0 \end{pmatrix}, \\
        & \gamma\left(y_{I4}\right) = \begin{pmatrix} 0 & \phi(x_j x_i) \\ \phi(x_i x_j) & 0 \end{pmatrix}, 
        & \gamma\left(y_{I5}\right) = \begin{pmatrix} \phi(x_j x_i x_j) & 0 \\ 0 & \phi(x_i) \end{pmatrix}, \\
        & \gamma\left(y_{I6}\right) = \begin{pmatrix} \phi(x_j x_k) & 0 \\ 0 & \Id \end{pmatrix}, \text{ and} 
        & \gamma\left(y_{I7}\right) = \begin{pmatrix} \phi(x_j) & 0 \\ 0 & \phi(x_k) \end{pmatrix} 
    \end{align*}
    for all $I = (i,j,k) \in \mcC$. It is straightforward to show that
    $\gamma$ is a $C\eps$-representation of $\Gamma$, where $C$ is a positive
    constant $\leq 15$. For instance, consider the relation $y_{I5}^2 = e$. To
    show that $\gamma(y_{I5})^2 \approx \Id$, we need to show that $\phi(x_j
    x_i x_j)^2 \approx \Id$.  Write $X \approx_{\eps} Y$ to mean that
    $\norm{X-Y} \leq \eps$. Since $\phi(x_i)^2 \approx_{\eps} \Id$ and
    $\phi(x_j)^2 \approx \Id$, we have $\phi(x_i x_k x_i)^2 \approx_{3 \eps} \Id$.
    We can conclude from this that $\gamma(y_{I5})^2 \approx_{3\eps} \Id$
    (we can do slightly better by averaging over the blocks of
    $\gamma(y_{I5})$, but we ignore this to simplify the analysis). We can
    similarly show that $\gamma(y_{Ij})^2 \approx_{3\eps} \Id$ for all $1 \leq
    j \leq 7$, and that the linear relations in Equation
    \eqref{E:additionalrelations} hold to within $3 \eps$. 

    This leaves the commuting relations. Consider the relation $y_{I3} y_{I4}
    y_{I5} = e$. We want to show that $\gamma(y_{I3})$, $\gamma(y_{I4})$, and
    $\gamma(y_{I5})$ approximately commute. But since $\gamma(y_{I3})
    \gamma(y_{I4}) \gamma(y_{I5}) \approx_{3\eps} \Id$ and $\gamma(y_{Ij})^2
    \approx_{3\eps} \Id$, we conclude that 
    \begin{equation*}
        \gamma(y_{I4}) \gamma(y_{I5}) \approx_{3\eps} \gamma(y_{I3})^* \approx_{3\eps}
            \gamma(y_{I3}) \approx_{3\eps} \gamma(y_{I5})^* \gamma(y_{I4})^* 
            \approx_{6\eps} \gamma(y_{I5}) \gamma(y_{I4}),
    \end{equation*}
    or in other words, $\gamma(y_{I4}) \gamma(y_{I5}) \approx_{15\eps}
    \gamma(y_{I5}) \gamma(y_{I4})$. The other commuting relations follow
    similarly.
    
    Let $\Psi$ be the lift of $\psi$ sending $x_i \mapsto x_i$. Then $\gamma
    \circ \Psi = \phi \oplus \phi$. Once again, the same construction applies
    when $\psi$ is an exact representation, so $\psi$ is an $fa^*$-embedding
    by Lemma \ref{L:trick}.
\end{proof}
Note that if $j = k$ in a relation $x_i x_j x_i = x_k$, then the system in
Equation \eqref{E:additionalrelations} is precisely the Mermin-Peres magic
square \cite{Me90, Pe90}. The magic square has previously been used by Ji to
show that linear system games can require a (finite but) arbitrarily high
amount of entanglement to play perfectly \cite{Ji13}. 

The proof of Proposition \ref{P:linearplusconjugacy} has several interesting
features:
\begin{rmk}\label{R:embeddingsize}
    Let $G = \Gamma(A,b,\mcC)$ be an $m \times n$ linear-plus-conjugacy group,
    and let $\Gamma' = \Gamma'(A',b')$ be the solution group constructed in the
    proof of Proposition \ref{P:linearplusconjugacy}. Then, accounting for
    Lemma \ref{L:nice}, the system $A'x=b'$ has $11n + 8 c + 1$
    variables and $8n + m + 7c$ equations, where $c = |\mcC|$ is
    the number of conjugacy relations. A presentation for $\Gamma'$ can be
    constructed in polynomial time in $m$, $n$, and $c$. 

    The proofs of Lemma \ref{L:nice} and Proposition
    \ref{P:linearplusconjugacy} show that there is a constant $C > 0$, and a
    lift $\Psi$ of the homomorphism $G \arr \Gamma'$ to the defining free
    groups, such that for any $d$-dimensional $\eps$-representation $\phi$ of
    $G$, there is a $4d$-dimensional $C \eps$-representation $\psi$ of
    $\Gamma'$ with $\psi \circ \Psi = \phi^{\oplus 4}$. Taking into account the
    fact that we have to change the presentation of the group $K$ in the proof
    of Lemma \ref{L:nice}, we can take the constant $C \leq 75$. The lift
    $\Psi$ can be chosen to send the generators of $G$ to generators of
    $\Gamma'$ (although not every generator of $\Gamma'$ will lie in the image
    of $\Psi$). 
\end{rmk}

It is important for our argument that the $fa^*$-embedding in Proposition
\ref{P:linearplusconjugacy} is over $\Z_2$. However, we can go a little further
in what type of groups can be embedded if we drop this requirement.
\begin{defn}\label{D:homogeneous}
    Suppose $A$ is an $m \times n$ matrix over $\Z_2$, and $\mcC \subseteq
    [n] \times [n] \times [n]$. Let 
    \begin{align*}
        \HLPC(A,\mcC) : = \Big\langle x_1,\ldots,x_n :\; & x_n^2 = e \text{ for all } 1 \leq j \leq n, \\
                                               & \prod_{j=1}^n x_j^{A_{ij}} = e \text{ for all } 
                                                    1 \leq i \leq m, \\
                                               & x_j x_k = x_k x_j \text{ if } j,k \in V_i(A)
                                                    \text{ for some } 1 \leq i \leq m, \text{ and } \\ 
                                               & x_i x_j x_i = x_k \text{ for all } (i,j,k) \in \mcC
                                                \Big\rangle.
    \end{align*}
    We say that a group $G$ is a \emph{homogeneous-linear-plus-conjugacy group}
    if it has a presentation of this form. 
\end{defn}
Since $\HLPC(A,\mcC)$ is not presented over $\Z_2$, a
homogeneous-linear-plus-conjugacy group is \emph{not} a linear-plus-conjugacy
group. However, the two types of groups are closely related, as $\HLPC(A,\mcC)
\times \Z_2 = \Gamma(A,0,\mcC)$. 

\begin{defn}\label{D:extended}
    Suppose $A$ is an $m \times n$ matrix over $\Z_2$, $\mcC_0 \subseteq [n]
    \times [n] \times [n]$, $\mcC_1 \subseteq [\ell] \times [n] \times [n]$, and
    $L$ is an $\ell \times \ell$ lower-triangular matrix with non-negative integer
    entries.  Let
    \begin{align*}
        \EHLPC(A,\mcC_0,\mcC_1,L) : = \Big\langle \Gamma_0(A,\mcC_0),\; y_1,\ldots,y_\ell :\; & 
                                           y_i x_j y_i^{-1} = x_k \text{ for all } (i,j,k) \in \mcC_1, \text{ and}\\
                                           & y_i y_j y_i^{-1} = y_j^{L_{ij}} \text{ for all } i>j \text{ with } L_{ij} > 0
                                                \Big\rangle.
    \end{align*}
    We refer to the generators $x_i$ in this presentation as
    \emph{involutary generators}, and to the generators $y_j$ as
    \emph{non-involutary generators}.
    We say that a group $G$ is an \emph{extended
    homogeneous-linear-plus-conjugacy group} if it has a presentation of this form. 
\end{defn}

\begin{prop}\label{P:extended}
    Let $G = \EHLPC(A,\mcC_0,\mcC_1,L)$ as in Definition \ref{D:extended}, where
    $A$ is an $m \times n$ matrix. Then there is an $m \times n'$ matrix $A'$ and a
    set $\mcC' \subset [n'] \times [n'] \times [n']$, where $n \leq n'$,
    such that there is an $fa^*$-embedding $\psi : G \arr \HLPC(A',\mcC')$ with
    $\psi(x_i) = x_i$ for all $1 \leq i \leq n$. 
\end{prop}
\begin{proof}
    Suppose $G$ has $\ell$ non-involutary generators, and let
    \begin{equation*}
        G' = \langle G, z, w : z^2 = w^2 = e, y_1 = zw, zy_i = y_i z \text{ for } i=2,\ldots,\ell \rangle.
    \end{equation*}
    We claim that the natural morphism $\psi : G \arr G'$ is an
    $fa^*$-embedding.  Indeed, let $\Psi : \mcF(S) \arr \mcF(S \cup \{z,w\})$
    be the natural inclusion, where $S = \{x_1,\ldots,x_n,y_1,\ldots,y_\ell\}$.
    Given an $\eps$-representation $\phi$ of $G$, define an approximate
    representation $\gamma$ of $G'$ by
    \begin{align*}
       \gamma(x_i) & = \begin{pmatrix} \phi(x_i) & 0 \\ 0 & \Id \end{pmatrix}, & 
       \gamma(z) & = \begin{pmatrix} 0 & \Id \\ \Id & 0 \end{pmatrix}, \\
       \gamma(w) & = \begin{pmatrix} 0 & \phi(y_1)^* \\ \phi(y_1) & 0 \end{pmatrix}, &
       \gamma(y_1) & = \begin{pmatrix} \phi(y_1) & 0 \\ 0 & \phi(y_1)^* \end{pmatrix},  \text{ and} 
    \end{align*}
    \begin{equation*}
        \gamma(y_i) = \begin{pmatrix} \phi(y_i) & 0 \\ 0 & \phi(y_i) \end{pmatrix} \text{ for }
            i = 2,\ldots,\ell.
    \end{equation*}
    Because $L$ is lower-triangular, $G'$ has no defining relations of the form
    $y_1 y_i y_1^{-1} = y_i^{L_{1i}}$. Suppose $L_{i1} > 0$, so that $\phi(y_i) \phi(y_1)
    \phi(y_i)^* \approx_{\eps} \phi(y_1)^{L_{i1}}$, where once again $X \approx_{\eps} Y$
    means that $\norm{X - Y} \leq \eps$. Then $\phi(y_i) \phi(y_1)^*
    \phi(y_i)^* \approx_{\eps} \phi(y_1)^{-L_{i1}}$, so $\psi(y_i) \psi(y_1)
    \psi(y_i)^* \approx_{\eps} \psi(y_1)^{L_{i1}}$. It is easy to see that
    the remaining defining relations of $G'$ hold to within $\eps$, so $\psi$
    is an $\eps$-representation of $G'$. Since $\phi$ is a direct summand of
    $\gamma \circ \Psi$, we can apply Lemma \ref{L:trick} with $N=2$ and $C=1$
    to see that $\psi$ is an $fa$-embedding. The same construction
    for exact representations shows that $\psi$ is an $fa^*$-embedding.

    Next, observe that $G'$ is an extended homogeneous-linear-plus-conjugacy group with $\ell-1$
    non-involutary generators. Indeed, suppose $(1,j,k) \in \mcC_1$. Then the
    defining relation $y_1 x_j y_1^{-1} = x_k$ is equivalent to the relation $z
    w x_j w z = x_k$. By adding an
    ancilla variable $Z_{jk}$ with $Z_{jk}^2 = e$, we can replace this relation
    with the two conjugacy relations $w x_j w = Z_{jk}$ and $z Z_{jk} z = x_k$.
    Similarly, suppose $L_{i1} > 0$. Then the relation $y_i y_1 y_i = y_1^{L_{i1}}$ is
    equivalent to the relation $y_i w y_i^{-1} = w (zw)^{L_{i1}-1}$. Once
    again, we can replace this relation with a sequence of conjugacy relations
    by adding ancilla variables. For instance, if $L_{i1}=3$, then we would add
    ancilla variables $W_{i0}$ and $W_{i1}$ with $W_{i0}^2 = W_{i1}^2 = e$, and
    conjugacy relations $z w z = W_{i0}$, $w W_{i0} w = W_{i1}$, and $y_i w
    y_{i}^{-1} = W_{i1}$. After making these replacements, the only relation
    containing $y_1$ is $y_1 = zw$, so we can remove $y_1$ from the set of
    generators. The commuting relations added in $G'$ are equivalent to $y_i z
    y_i^{-1} = z$ for all $2 \leq i \leq \ell$, so $G'$ is an extended
    homogeneous-linear-plus-conjugacy group. The additional variables (including the ancilla)
    are involutary generators, so $G'$ has $\ell-1$ non-involutary generators. 

    Iterating this construction, we get a sequence of $fa^*$-embeddings terminating
    in a homogeneous-linear-plus-conjugacy group, as desired. 
\end{proof}
The reason the above argument does not apply for groups over $\Z_2$ is that,
if we set $\gamma(J) = \phi(J) \oplus \Id$, then $\gamma(J)$ would not commute with 
$\gamma(z)$ and $\gamma(w)$, while if we set $\gamma(J) = \phi(J) \oplus \phi(J)$,
then any linear relations containing $J$ would not be satisfied. 
\begin{rmk}\label{R:embeddingsize2}
    The above proof shows that, in Proposition \ref{P:extended}, we can take
    \begin{equation*}
        n' = n + 2 \ell + \binom{\ell}{2} + |\mcC_1| + \operatorname{sum}(L)
    \end{equation*}
    and
    \begin{equation*}
        |\mcC'| = |\mcC_0| + 2 |\mcC_1| + 2 \binom{\ell}{2} + \operatorname{sum}(L) + \#(L),
    \end{equation*}
    where $\ell$ is the number of non-involutary generators, $\operatorname{sum}(L)$ is
    the sum of the entries of $L$, and $\#(L)$ is the number of non-zero entries
    of $L$. The matrix $A'$ and set $\mcC'$ can be constructed in polynomial time
    in $m$, $n$, $\ell$, $|\mcC_0|$, $|\mcC_1|$, and $\operatorname{sum}(L)$. 
\end{rmk}

\section{Proof of Theorem \ref{T:main}}

The point of this section is to prove the following proposition, and hence
finish the proof of Theorem \ref{T:main}.
\begin{prop}\label{P:main}
    There is a solution group $\Gamma$ for which $J$ is trivial in
    finite-dimensional representations, but non-trivial in finite-dimensional
    approximate representations.
\end{prop}
For the proof of Proposition \ref{P:main}, it is convenient to work with sofic
groups. We do not need to know the definition of soficity, just that the class
of sofic groups has the following properties:
\begin{enumerate}[(1)]
    \item Amenable groups are sofic.
    \item Sofic groups are hyperlinear. 
    \item If $H$ is an amenable subgroup of a sofic group $G$, and $\alpha : H \arr G$ is
        injective homomorphism, then the HNN extension of $G$ by $\alpha$ is sofic.
\end{enumerate}
An expository treatment of sofic groups can be found in \cite{CL15}. In
particular, the last ``closure property'' can be found in \cite[Section
II.4]{CL15}. 

We need one more general-purpose lemma before proceeding to the proof. 
\begin{lemma}\label{L:homogextension}
    Suppose $G = \langle S : R \rangle$ is a finitely-presented group, where
    $R$ contains the relation $a^2 = e$ for some $a \in S$. Let
    \begin{equation*}
        \widehat{G} := \langle G , t : t^2 = e, tat = Ja \rangle_{\Z_2},
    \end{equation*}
    where $J, t \not\in S$. If $a$ is non-trivial in
    approximate representations of $G$, then $J$ is non-trivial in
    approximate representations of $\widehat{G}$. 
\end{lemma}
Note that $\widehat{G}$ is the ``$\Z_2$-HNN extension'' of $G \times
\Z_2$, where $J$ is the generator of the $\Z_2$ factor, by the order-two
automorphism sending $a \mapsto Ja$ and $J \mapsto J$. 
\begin{proof}
    For the purposes of this proof, if $X$ is a linear operator
    on a finite-dimensional Hilbert space $H$, let $\widetilde{\tr}(X) :=
    \tr(X) / \dim H$. Suppose $\phi$ is an $\eps$-representation of $G$
    with $\phi(a)^2 = 1$ and $\tr(\phi(a)) \geq 0$. Because the eigenvalues
    of $\phi(a)$ belong to $\{\pm 1\}$, we can choose a basis so that
    $\phi(a) = \Id_{d_0} \oplus (-\Id)_{d_0} \oplus \Id_{d_1}$, where $d_1 =
    \tr(\phi(a))$. Define an approximate representation $\psi$ of $\widehat{G}$ by 
    \begin{equation*}
        \psi(x) = \phi(x) \text{ for all } x \in S,
        \psi(J) = -\Id, \text{ and } 
        \psi(t) = \begin{pmatrix} 0 & \Id_{d_0} & 0 \\
                                   \Id_{d_0} & 0 & 0 \\ 
                                   0 & 0 & \Id_{d_1} \end{pmatrix}.
    \end{equation*}
    Clearly $\norm{\psi(r) - \Id} = \norm{\phi(r) - \Id} \leq \eps$ for all
    relations $r \in R$, $\psi([J,s]) = \Id$ for all $s \in S \cup \{t\}$, and
    $\psi(t)^2 = \psi(J)^2 = \Id$. For the remaining relation, 
    \begin{equation*}
        \norm{\psi(t a t) - \psi(J a)} = \norm{0_{2d_0} \oplus 2\Id_{d_1}} 
            = 2\sqrt{\frac{d_1}{2d_0 + d_1}} = 2\sqrt{\widetilde{\tr}(\phi(a))}.
    \end{equation*}
    So $\psi$ will be a $\max\left(\eps, 2\sqrt{\widetilde{\tr}(\phi(a))}\right)$-representation 
    with $\norm{\psi(J)-\Id} = 2$. 
         
    To make $\widetilde{\tr}(\phi(a))$ small, we can use the tensor-power
    trick as in Section II.2 of \cite{CL15}. Suppose $a$ is non-trivial
    in approximate representations of $G$. By Lemmas \ref{L:smallchanges} and \ref{L:stability1}, there
    is a constant $\delta > 0$, such that for all $\eps > 0$, there is an
    $\eps$-representation $\phi$ of $G$ with $\norm{\phi(a) - \Id} > \delta$
    and $\phi(a)^2 = \Id$.  Given $\eps > 0$, find an integer $k$ such that
    \begin{equation*}
        \left(1 - \frac{\delta^2}{4} \right)^{k} \leq \frac{\eps^2}{4},
    \end{equation*}
    and let $\phi$ be an $\frac{\eps}{k}$-representation with $\norm{\phi(a)
    - \Id} > \delta$ and $\phi(a)^2 = \Id$. Suppose $\phi$ has dimension $d$,
    and let $\gamma$ be the direct sum of $\phi$ with $d$ copies of the trivial
    representation. Then $\gamma$ is an $\frac{\eps}{k}$-representation of
    $G$ by Lemma \ref{L:sumandproduct}, and furthermore $\gamma(a)^2 = \Id$, 
    $\tr(\gamma(a)) = d + \tr(\phi(a)) \geq 0$, and
    \begin{equation*}
        \norm{\gamma(a) - \Id} = \norm{ \Id_{d} \oplus \phi(a) - \Id_{2d} }
            = \frac{1}{\sqrt{2}} \norm{ \phi(a) - \Id}
            > \frac{\delta}{\sqrt{2}}.
    \end{equation*}
    Since $\gamma(a)$ is self-adjoint,
    \begin{equation*}
        \norm{\gamma(a) - \Id}^2 = 2 - 2 \widetilde{\tr}(\gamma(a)), 
    \end{equation*}
    so we conclude that
    \begin{equation*}
        0 \leq \widetilde{\tr}(\gamma(a)) \leq 1 - \frac{\delta^2}{4}. 
    \end{equation*}
    Since $\widetilde{\tr}(X^{\otimes k}) = \widetilde{\tr}(X)^k$, Lemma
    \ref{L:sumandproduct} implies that $\gamma^{\otimes k}$ is an
    $\eps$-representation of $\widehat{G}$ with 
    \begin{equation*}
        0 \leq \widetilde{\tr}(\gamma^{\otimes k}(a)) \leq \left(1 - \frac{\delta^2}{4}\right)^k \leq \frac{\eps^2}{4}.
    \end{equation*}
    Applying the argument of the first paragraph to $\gamma^{\otimes k}$, we get an
    $\eps$-representation $\psi$ of $G$ with $\norm{\psi(J)-\Id} = 2$. This
    shows that $J$ is non-trivial in approximate representations of $\widehat{G}$. 
\end{proof}
We are now ready to prove Proposition \ref{P:main}.  Note that any hyperlinear
but non-residually-finite group has an element which is trivial in
finite-dimensional representations, but non-trivial in approximate
representations. To prove Proposition \ref{P:main}, we show that
\begin{equation*}\label{L:K1}
    K = \langle x,y,a,b : a^2 = b^2 = e, ab = ba, y a y^{-1} = a, y b
                y^{-1} = ab, x y x^{-1} = y^2\rangle
\end{equation*}
is an extended homogeneous-linear-plus-conjugacy group which is hyperlinear
but non-residually finite.  Indeed, to see that $K$ has a presentation as in
Definition \ref{D:extended}, we can introduce a third variable $c$ with $c^2 =
e$ and $c = ab$. Then $K$ is equivalent to the extended
homogeneous-linear-plus-conjugacy group with three involutary generators
$a,b,c$, one linear relation $abc = e$ (along with the corresponding commuting
relations), two non-involutary generators $x$ and $y$, and three conjugacy
relations $y a y^{-1} = a$, $y b y^{-1} = c$, and $x y x^{-1} = y^2$. 
For the remainder of this section, $K$ will refer to this group. 
\begin{lemma}\label{L:main1}
    $K$ is sofic, and the element $a \in K$ is non-trivial.
\end{lemma}
\begin{proof}
    $K_1 := \langle y,a,b: a^2 = b^2 = e, ab = ba, y a y^{-1} = a, yby^{-1} = ab \rangle$
    is isomorphic to $\Z \ltimes (\Z_2 \times \Z_2)$, and in particular is
    solvable (hence amenable). The group $K$ is the HNN extension of $K_1$ by
    the injective endomorphism of $\langle y \rangle \iso \Z$ sending $y
    \mapsto y^{2}$. Hence $K$ is sofic by properties (1) and (3) of sofic groups above. 
    In addition, the natural morphism $K_1 \arr K$ is injective. Since $a$ is
    clearly non-trivial in $K_1$, we conclude that $a$ is non-trivial in $K$. 
\end{proof}

The following lemma comes from discussions with Tobias Fritz.
\begin{lemma}\label{L:main2}
    The element $a \in K$ is trivial in all finite-dimensional representations
    of $K$.
\end{lemma}
\begin{proof}
    By a theorem of Mal'cev \cite{Ma65}, it suffices to show that $a$ is
    trivial in finite representations, rather than finite-dimensional
    representations. So let $\phi : G \arr H$ be a homomorphism from $G$ to a
    finite group $H$. Now the order $k$ of $\phi(x)$ is finite, so $\phi(y) =
    \phi(x)^k \phi(y) \phi(x)^{-k} = \phi(y)^{2^k}$. It follows that the order $m
    = |\phi(y)|$ of $\phi(y)$ divides $2^k-1$, and in particular is odd. Since
    $\phi(y) \phi(b) \phi(y)^{-1} = \phi(ab)$ and $\phi(y) \phi(ab)
    \phi(y)^{-1} = \phi(b)$, we conclude that $\phi(b) = \phi(y)^m \phi(b)
    \phi(y)^{-m} = \phi(ab)$. Consequently $\phi(a) = \Id$ as desired. 
\end{proof}

\begin{proof}[Proof of Proposition \ref{P:main}]
    By Proposition \ref{P:extended}, there is an $fa$-embedding of $K$ to a
    homogeneous-linear-plus-conjugacy group $G = \HLPC(A,\mcC)$, in which $a
    \in K$ is mapped to a generator $x_i$ of $G$. Let 
    \begin{equation*}
        \widehat{G} = \langle G, t : t^2 = e, t x_i t = J x_i \rangle_{\Z_2}.
    \end{equation*}
    The relation $t x_i t = J x_i$ can be replaced with the relations $t x_i t
    = Z$ and $Z x_i = J$, where $Z$ is an ancilla variable with $Z^2 = e$. With
    this presentation, $\widehat{G}$ is a linear-plus-conjugacy group. By Proposition
    \ref{P:linearplusconjugacy}, there is an $fa$-embedding over $\Z_2$ of
    $\widehat{G}$ to a solution group $\Gamma$. 
    
    By Lemma \ref{L:main1}, $a$ is non-trivial in approximate representations of
    $K$, and hence $x_i$ is non-trivial in approximate representations of $G$. 
    By Lemma \ref{L:homogextension}, $J_{\widehat{G}}$ is non-trivial in
    approximate representations of $\widehat{G}$, and we conclude that
    $J_{\Gamma}$ is non-trivial in approximate representations of $\Gamma$.

    Finally, there is a morphism from $K$ to $\widehat{G}$ which sends $a$ to
    $x_i$, so $x_i$ will be trivial in all finite-dimensional representations
    of $\widehat{G}$ by Lemma \ref{L:main2}. But since $J_{\widehat{G}} =
    [t,x_i]$, this means that $J_{\widehat{G}}$ (and hence $J_{\Gamma}$) is
    trivial in all finite-dimensional representations of $\widehat{G}$. 
\end{proof}

\begin{proof}[Proof of Theorem \ref{T:main}]
    Let $\Gamma$ be the solution group from Proposition \ref{P:main}, and let
    $\mcG$ be the associated game. Since $J$ is trivial in finite-dimensional
    representations, Theorem \ref{T:solutiongroup} implies that $\mcG$ does not
    have a perfect strategy in $C_{qs}$. But since $J$ is non-trivial in
    approximate representations, Proposition \ref{P:reptostrat} implies that
    $\mcG$ has a perfect strategy in $C_{qa}$.
\end{proof}

\begin{rmk}
    By Remarks \ref{R:embeddingsize} and \ref{R:embeddingsize2}, the linear
    system constructed in the proof of Theorem \ref{T:main} will have $235$
    variables and $184$ equations. 
\end{rmk}

\section{Proofs of Theorems \ref{T:main2} and \ref{T:main3}}

To prove Theorem \ref{T:main2}, we want to find a hyperlinear group with an
undecidable word problem, which $fa$-embeds in a solution group. For Theorem
\ref{T:main3}, we want to find a family of residually finite groups with arbitrarily hard
(albeit computable) word problems, which $fin$-embed in solution groups.
Fortunately, such groups are provided by Kharlampovich \cite{Kha81} and
Kharlampovich, Myasnikov, and Sapir \cite{KMS17}. Since the presentations
are rather complicated, we do not repeat them here. Instead, we summarize
some points of the construction from \cite{KMS17} in the following theorem.

It is helpful to use the following notation: given $S_0 \subseteq S_1$, let
$\mcN(S_0,S_1)$ denote the normal subgroup generated by $S_0$ in the free group $\mcF(S_1)$.
Note that if $S_1 \subseteq S$, then $\mcN(S_0,S_1)$ is a (not necessarily
normal) subgroup of $\mcF(S)$ in a natural way. Also, if $x,y$ are group
elements, recall that $[x,y] = xyx^{-1}y^{-1}$, and $x^y = y x
y^{-1}$.\footnote{This is the reverse of the convention in \cite{KMS17}, where $[x,y] =
x^{-1} y^{-1} x y$ and $x^y = y^{-1} x y$.}
\begin{theorem}[\cite{KMS17}, see also \cite{Kha81}]\label{T:kha}
    Let $X \subseteq \mbN$ be recursively enumerable. Then there is a
    finitely-presented solvable group $K_X = \langle S : R \rangle$ with
    the following properties:
    \begin{enumerate}[(1)]
        \item The set $S$ is divided into three subsets $L_i$, $i=0,1,2$. 
        \item The relations in $R$ come in three types:
            \begin{enumerate}[(a)]
                \item $R$ contains the relations $x^2 = e$ for all $x \in L_0 \cup L_1$.
                \item $R$ also contains commuting relations of the form
                    $xy = yx$, for certain pairs $x,y \in S$.
                \item Every other relation $r \in R$ belongs to some normal subgroup
                    $\mcN(S_0,S_1)$, where $S_1 \subseteq S$ and $S_0 \subseteq
                    (L_0 \cup L_1) \cap S_1$ are such that the image of
                    $\mcN(S_0,S_1)$ in $K_X$ is abelian.
            \end{enumerate}
        \item The image of $\mcN(L_0,S)$ in $K_X$ is abelian.
        \item There are elements $z_0,z_1 \in L_0$, $A_1,A_2  \in L_1$, and $a,a' \in L_2$, such that
            $n \in X$ if and only if 
            \begin{equation*}
                [A_2,[A_1,w(2^n)]] = [A_2,[A_1,z_0]]
            \end{equation*}
            in $K_X$, where $w(m)$ is defined by
            \begin{equation*}
                w(m) := \begin{cases} 
                                z_1 & m = 0 \\
                                w(m-1) w(m-1)^{a^{-1}} w(m-1)^{a} w(m-1)^{a'} & m \geq 1 \\
                                \end{cases}.
            \end{equation*}
        \item If $X$ is recursive, then $K_X$ is residually finite.
    \end{enumerate}
\end{theorem}
Note that there is some overlap between relations of type (2b) and (2c).
Indeed, if $[x,y]=e$ is a relation, then the image of $\mcN(\{x\},\{x,y\})$ in
$K_X$ is equal to $\langle x \rangle$, and in particular is abelian. Since
$[x,y]$ belongs to $\mcN(\{x\},\{x,y\})$, any relation $[x,y] = e$ of type (2b)
with $x \in L_0 \cup L_1$ can also be regarded as a relation of type (2c).

To see that property (4) of the theorem holds from the description in
\cite{KMS17}, it is helpful to note that, by properties (1), (2a), and (3) of the
theorem, $w(m)$ is an involution for all $m \geq 1$.

\begin{lemma}\label{L:kha}
    Suppose $K = \langle S : R \rangle$ is a finitely-presented group
    satisfying properties (1) and (2) of Theorem \ref{T:kha}.  Then $K$ is an
    extended homogeneous-linear-plus-conjugacy group (as in Definition
    \ref{D:extended}). 

    Furthermore, if $S_0 \subseteq S_1 \subseteq S$ are two subsets such that
    $S_0 \subseteq L_0 \cup L_1$, and the image of $\mcN(S_0,S_1)$ in $K$ is
    abelian, then for every $w \in \mcN(S_0,S_1)$, there is a presentation of
    $K$ as an extended homogeneous-linear-plus-conjugacy group in which $w$ is
    equal in $K$ to one of the involutary generators $x_i$. 
\end{lemma}
\begin{proof}
    The generating set of $K$ is split into involutary generators $L_0 \cup
    L_1$ and non-involutary generators $L_2$.  Since the order on
    non-involutary generators matters in Definition \ref{D:extended}, choose an
    arbitrary enumeration $y_1,\ldots,y_k$ of $L_2$. According to property (2)
    of Theorem \ref{T:kha}, the defining relations for $K$ (aside from the
    involutary relations on $L_0 \cup L_1$) fall into two types: (2b) and (2c).
    Both types of relations can be rewritten as linear and conjugacy relations
    of the types allowed in Definition \ref{D:extended}. 
    Indeed, commuting relations (relations of type (2b)) can be regarded as
    conjugacy relations (note that for relations $y_i y_j = y_j y_i$, we can
    choose either $y_i y_j y_i^{-1} = y_j$ or $y_j y_i y_j^{-1} = y_i$
    depending on whether $i > j$ or $i < j$).  

    This leaves relations of type (2c). For this, we first prove the second
    part of the lemma. Suppose that the image of $\mcN(S_0,S_1)$ is
    abelian in $K$, where $S_0 \subset L_0 \cup L_1$. We claim that for any
    non-trivial element $w \in \mcN(S_0,S_1)$, there is a finite set of
    generators $S_w$ and relations $R_w \subset \mcF(L_0 \cup L_1 \cup L_2 \cup
    S_w)$ such that 
    \begin{enumerate}[(i)]
        \item $R_w$ consists of linear and conjugacy relations as in Definition
            \ref{D:extended}, 
        \item the relations 
            \begin{equation*}
                \widetilde{R}_w := R_w \cup \{ s^2 = e : s \in L_0 \cup L_1 \cup S_w\}
            \end{equation*}
            imply that $w$ is equal to an element of $L_0 \cup S_w$, and
        \item the added generators $S_w$ and relations $R_w$ do not change the
            group, i.e. the inclusion
            \begin{equation*}
                K \arr \langle K, S_w : R_w \cup \{ s^2 = e : s \in S_w \} \rangle
            \end{equation*}
            is an isomorphism.
    \end{enumerate}
    To prove the claim, we use induction on the length of $w$ in $\mcF(S_1)$.
    The claim is trivially true if $w \in S_0^{\pm}$. Suppose $w = z x z^{-1}$,
    where $x \in \mcN(S_0)$ has length less than $w$, and $z \in
    S_1$. By induction, there is a set of ancilla variables
    $S_{x}$ and relations $R_x$ satisfying properties (i)-(iii) for $x$. In
    particular, the relations $\widetilde{R}_x$ imply that $x$ is equal to some
    $X \in S_0 \cup S_x$. Then we can set $S_w := S_x \cup \{W\}$, where $W$ is
    a new indeterminate, and $R_w := R_x \cup \{W = z X z\}$ or $R_x \cup \{W =
    z X z^{-1}\}$ depending on whether $z \in L_0 \cup L_1$ or $z \in L_2$. If
    $w = z^{-1} x z$, then we do the same thing, but using $z W z^{-1} = X$ in
    place of $W = z X z^{-1}$.  Finally, suppose that $w = x_1 \cdots x_n$,
    where each $x_i \in \mcN(S_0,S_1)$ has smaller length than $w$. By induction, there
    are sets $S_{x_i}$ and relations $R_{x_i}$ implying that $x_i$ is equal to
    some $X_i \in L_0 \cup S_{x_i}$.  We then set $S_w := \bigcup S_{x_i} \cup
    \{W\}$, where $W$ is again a new indeterminate, and 
    \begin{equation*}
        R_w := \bigcup R_{x_i} \cup \{W X_1 \cdots X_n = e = [W,X_i] = [X_i, X_j] \text{ for all } 1 \leq i,j \leq n\}. 
    \end{equation*}
    Since the image of $\mcN(S_0,S_1)$ in $K$ is abelian, adding the relations
    $R_w$ does not change $K$. This proves the claim. 

    Now suppose that $K$ has a defining relation in $\mcN(S_0,S_1)$. If $r = z x z^{-1}$
    for some $x \in \mcN(S_0,S_1)$ and $z \in S_1^{\pm}$, then $r$ can be replaced
    with the simpler relation $x$. So we can assume without loss of generality
    that $r = x_1 \cdots x_n$, where each $x_i \in \mcN(S_0,S_1)$. By the
    claim, we can add ancilla variables and relations such that each $x_i$ is
    equal to an involutary generator $X_i$ in $K$, and the relation $r$ can be
    replaced with the linear relation $X_1 \cdots X_n = e$. We conclude that
    $K$ is an extended homogeneous-linear-plus-conjugacy group. The claim also
    immediately implies the second part of the lemma. 
\end{proof}

We now come to the main result of this section.
\begin{prop}\label{P:kha}
    Let $X \subseteq \mbN$ be a recursively enumerable set. Then there
    is a family of solution groups $\Gamma_n = \Gamma\left(A^{(n)},
    b^{(n)}\right)$, $n \geq 1$, such that
    \begin{enumerate}[(a)]
        \item $A^{(n)} x = b^{(n)}$ is an $\exp(O(n)) \times \exp(O(n))$ linear
            system;
        \item the function $n \mapsto \left(A^{(n)},b^{(n)}\right)$ is computable
            in $\exp(O(n))$-time;
        \item $J_{\Gamma_n}$ is non-trivial in $\Gamma_n$ if and only if
            $n \in X$; 
        \item if $J_{\Gamma_n}$ is non-trivial in $\Gamma_n$, then $J_{\Gamma_n}$ is
            non-trivial in approximate representations; and
        \item if $X$ is recursive and $J_{\Gamma}$ is non-trivial in 
            $\Gamma_n$, then $J_{\Gamma_n}$ is non-trivial in finite-dimensional
            representations.
    \end{enumerate}
\end{prop}
Before giving the proof, we need the following exact version of Lemma
\ref{L:homogextension}.
\begin{lemma}\label{L:finitefreeprod}
    Suppose $G = \langle S : R \rangle$ is a finitely-presented group, where $R$
    contains the relation $a^2 = e$ for some $a \in S$. Let
    \begin{equation*}
        \widehat{G} := \langle G, t : t^2 = e, tat = Ja \rangle_{\Z_2},
    \end{equation*}
    where $J,t \not\in S$. If $a$ is non-trivial in finite-dimensional representations
    of $G$, then $J$ is non-trivial in finite-dimensional representations of $\widehat{G}$.
\end{lemma}
\begin{proof}
    Suppose $a$ is non-trivial in finite-dimensional representations of $G$.
    A theorem of Baumslag states that the free product of two residually finite
    groups amalgamated over a finite subgroup is residually finite \cite{Ba63}. 
    Let $\widetilde{G} := G \times \Z_2$, where the generator of the $\Z_2$ factor
    is denoted by $J$, and let $H = \langle t, a : t^2 = a^2 = e, tat = aJ \rangle_{\Z_2}
    \iso \Z_2 \ltimes \Z_2 \times \Z_2$.  Then $\widehat{G}$ is isomorphic to amalgamated
    free product of $\widetilde{G}$ and $H$ over $\langle a,J \rangle \iso \Z_2
    \times \Z_2$, a finite group. While $\widetilde{G}$ is not necessarily
    residually finite, the group $\widetilde{G}^{fin}$ is residually finite by definition,
    and there is natural map from $\widehat{G}$ to the amalgamated free product
    of $\widetilde{G}^{fin}$ and $H$ over $\Z_2 \times \Z_2$. The image of
    $J_{\widetilde{G}}$ is non-trivial in $\widetilde{G}^{fin}$, and hence
    in the amalgamated product of $\widetilde{G}^{fin}$ and $H$. So $J$
    is non-trivial in finite-dimensional representations of $\widehat{G}$ by
    Baumslag's result.
\end{proof}

\begin{proof}[Proof of Proposition \ref{P:kha}]
    Given a recursively enumerable subset $X \subseteq \mbN$, let $K_X =
    \langle S : R \rangle$ be the associated group from Theorem \ref{T:kha}.
    Using the notation from property (4) of Theorem \ref{T:kha}, let $c(n) =
    [A_2,[A_1,w(2^n)]] [A_2,[A_1,z_0]]^{-1}$, so that $c(n) = e$ in $K_X$ if
    and only if $n \in X$. Since $c(n)$ belongs to $\mcN(L_0,S)$, Lemma
    \ref{L:kha} and property (3) of Theorem \ref{T:kha} implies that $K_X$ has
    a presentation $\langle S_n : R_n \rangle$ as an extended
    homogeneous-linear-plus-conjugacy group, in which $c(n)$ is equal to some
    involutary generator in $S_n$. Since the presentation $\langle S : R\rangle$
    is fixed, the size of $\langle S_n : R_n \rangle$ depends only on the number
    of ancilla generators and relations needed to set $c(n)$ equal to one of
    the involutary generators. Inspection of the argument from Lemma
    \ref{L:kha} reveals that we need to add $4m$ ancilla generators and
    relations to set $w(m)$ to an involutary generator. Thus $S_n$ and $R_n$
    will have size $O(2^n)$, and the function $n \mapsto (S_n,R_n)$ can be
    computed in $O(2^n)$-time. 

    By Proposition \ref{P:extended}, there is an $fa^*$-embedding from $\langle
    S_n : R_n \rangle$ to a homogeneous-linear-plus-conjugacy group $G_n$,
    in which $c(n)$ is mapped to some generator $x_i$. As in the proof of
    Proposition \ref{P:main}, let 
    \begin{equation*}
        \widehat{G}_n = \langle G_n, t : t^2 = e, t x_i t = J x_i \rangle_{\Z_2}. 
    \end{equation*}
    Then $\widehat{G}_n$ is a linear-plus-conjugacy group, and by Proposition
    \ref{P:linearplusconjugacy}, there is an $fa^*$-embedding of $\widehat{G}_n$
    in a solution group $\Gamma_n = \Gamma(A^{(n)}, b^{(n)})$. By Remarks
    \ref{R:embeddingsize} and \ref{R:embeddingsize2}, $A^{(n)}$ and $b^{(n)}$
    can be constructed in time polynomial in $|S_n|$ and $|R_n|$, so
    $A^{(n)}$ and $b^{(n)}$ satisfy parts (a) and (b) of the proposition.

    Suppose $c(n)$ is non-trivial. Since $K_X$ is solvable, it is hyperlinear,
    so $c(n)$ is non-trivial in approximate representations. By Lemma
    \ref{L:homogextension}, $J_{\Gamma_n}$ will be non-trivial in approximate
    representations. If $X$ is recursive, then $K_X$ will be residually finite
    by property (5) of Theorem \ref{T:kha}, and hence $J_{\Gamma_n}$ will be
    non-trivial in finite-dimensional representations by Lemma \ref{L:finitefreeprod}
    (this uses the fact that $fa^*$-embeddings are also $fin$-embeddings).
    On the other hand, if $c(n)$ is trivial then $J_{\Gamma_n}$ will be trivial.
    Hence parts (c)-(e) of the proposition follow from property (4) of Theorem
    \ref{T:kha}. 
\end{proof}

\begin{proof}[Proof of Theorem \ref{T:main2}]
    Let $X \subseteq \mbN$ be a recursively enumerable but non-recursive set,
    and take the family $\{\mcG_n : n \in \mbN \}$ of games associated to the
    solution groups $\{ \Gamma_n : n \in \mbN \}$ constructed in Proposition
    \ref{P:kha}. By Theorem \ref{T:solutiongroup2} and part (c) of Proposition
    \ref{P:kha}, $\mcG_n$ will have a perfect strategy in $C_{qc}$ if and only
    if $n \in X$. By Proposition \ref{P:reptostrat} and part (d) of Proposition
    \ref{P:kha}, $\mcG_n$ will have a perfect strategy in $C_{qc}$ if and only
    if it has a perfect strategy in $C_{qa}$. Because the function $n \mapsto
    \mcG_n$ is computable by part (b) of Proposition \ref{P:kha}, it is undecidable
    to determine if  the games in this family have perfect srategies in
    $C_{qa}$.
\end{proof}

\begin{proof}[Proof of Theorem \ref{T:main3}]
    Given a computable function $f(n)$, let $X \subseteq \mbN$ be a 
    recursive subset such that for any Turing machine accepting $X$,
    the running time over inputs $n \leq N$ is at least $f(N)$ when
    $N$ is sufficiently large.\footnote{Often when talking about the
    running time, we look at the maximum running time over inputs
    of size $\leq N$, rather than value $\leq N$. However, thinking of
    the running time in terms of the values of the inputs does not change
    the fact that such sets $X$ exist.} Once again, we can take
    the family of games $\{\mcG_n : n \in \mbN\}$ associated to the solution groups
    $\{\Gamma_n : n \in \mbN\}$ from Proposition \ref{P:kha}. Then part (a) of Theorem \ref{T:main3}
    follows from parts (a) and (b) of Proposition \ref{P:kha}, while parts (b)
    and (c) of Theorem \ref{T:main3} follow from parts (c) and (e) of Proposition
    \ref{P:kha}, as well as Theorems \ref{T:solutiongroup} and
    \ref{T:solutiongroup2}.
\end{proof}

\begin{proof}[Proof of Corollary \ref{C:main3}]
    Suppose there is an algorithm to decide if a linear system game has a
    perfect strategy in $C_q$. Let $g(n)$ be the running time of this algorithm
    on games coming from linear systems with at most $n$ rows and columns.
    Note that $g(n)$ is an increasing function.  Let $f(n)$ be any computable
    function such that 
    \begin{equation*}
        f(n) > g\left(2^{n^2}\right) + 2^{n^2}
    \end{equation*}
    for all $n \geq 1$. Let $\mcG_n$ be the family of games associated
    to $f(n)$ as in Theorem \ref{T:main3}. Then there is a constant $C$
    such that $\mcG_n$ has size $\leq 2^{Cn}$ for all $n \geq 1$, and the
    function $n \mapsto \mcG_n$ is computable in time $2^{Cn}$. Plugging
    $\mcG_n$ into the algorithm to decide whether a linear system game
    has a perfect strategy in $C_q$, we get an algorithm for the language
    \begin{equation*}
        X = \{n \in \mbN : \mcG_n \text{ has a perfect strategy in } C_q \}
    \end{equation*}
    with running time at most $g(2^{CN}) + 2^{CN}$ on inputs $n \leq N$.
    But by part (b) of Theorem \ref{T:main3}, when $N$ is sufficiently large
    the maximum running time on inputs $n \leq N$ for any algorithm for $X$
    must be at least $f(N)$. Since $N^2$ will eventually be larger than $C N$,
    we get a contradiction. Thus there is no algorithm to decide if a linear
    system game has a perfect strategy in $C_q$.
\end{proof}

\bibliographystyle{amsalpha}
\bibliography{approx}

\end{document}

%% file: header/ams_header.tex
\usepackage{amsmath,amsthm,amsfonts,amssymb,latexsym,verbatim,bbm,hyperref}
\usepackage[shortlabels]{enumitem}
\usepackage[top=1.5in,bottom=1.5in,left=1.5in,right=1.5in]{geometry}
\usepackage{subcaption}
\usepackage{afterpage}

\DeclareRobustCommand{\SkipTocEntry}[5]{}

\allowdisplaybreaks

\setlist{itemsep=.5\baselineskip,topsep=.5\baselineskip}

\numberwithin{equation}{section}
\theoremstyle{plain}
\newtheorem{theorem}{Theorem}[section]
\newtheorem{thm}[theorem]{Theorem}
\newtheorem{lemma}[theorem]{Lemma}

\newtheorem{prop}[theorem]{Proposition}

\newtheorem{defn}[theorem]{Definition}

\newtheorem{cor}[theorem]{Corollary}

\newtheorem{rmk}[theorem]{Remark}

\newcommand{\iso}{\cong}
\newcommand{\arr}{\rightarrow}

\newcommand{\R}{\mathbb{R}}
\newcommand{\C}{\mathbb{C}}
\newcommand{\Z}{\mathbb{Z}}

\newcommand{\eps}{\epsilon}

\newcommand{\Id}{\mathbbm{1}}

\newcommand{\mcB}{\mathcal{B}}
\newcommand{\mcC}{\mathcal{C}}

\newcommand{\mcF}{\mathcal{F}}
\newcommand{\mcG}{\mathcal{G}}

\newcommand{\mcI}{\mathcal{I}}

\newcommand{\mcN}{\mathcal{N}}
\newcommand{\mcO}{\mathcal{O}}

\newcommand{\mcU}{\mathcal{U}}

\newcommand{\mbN}{\mathbb{N}}
